\documentclass[2p]{elsarticle}
\usepackage{amssymb}
\usepackage{amsmath}
\usepackage{amsthm}
\newtheorem{definition}{Definition}[section]
\newtheorem{theorem}{Theorem}[section]
\newtheorem{lemma}{Lemma}[section]
\newtheorem{remark}{Remark}[section]
\usepackage{float}
\usepackage{graphicx}
\usepackage{epstopdf}

\begin{document}

\title{Soliton-like solutions of the Camassa--Holm equation \\ with variable coefficients and a small dispersion}

\author[inst1]{Yuliia Samoilenko}
\author[inst2]{Valerii Samoilenko}

\affiliation[inst1]{organization={Institut Camille Jordan},
addressline={43 Boulevard du 11 Novembre 1918},
city={Villeurbanne},
postcode={69622},
country={France}}

\affiliation[inst2]{organization={Institute of Mathematics of NAS of Ukraine},
addressline={3 Tereshchenkivs'ka Str.},
city={Kyiv},
postcode={01024},
country={Ukraine}}


%
\begin{abstract}
The paper deals with the Camassa--Holm equation with variable coefficients (vcCH equation) that is a direct generalization of the well-known Camassa--Holm equation. We focus on the mathematical description of particular solutions of the vcCH equation with a small dispersion that exhibit properties analogous to those of classical soliton and peakon solutions, and consider the construction of soliton- and peakon-like solutions in the form of asymptotic expansions, including both one-phase and two-phase cases.

The solution is expressed as the sum of a regular background common to both soliton- and peakon-like solutions and a singular component that captures their distinctive features, with the precise definition of the main singular term playing a central role. In the one-phase case, this term is determined, and the solvability of higher-order singular corrections is established in appropriate functional spaces, enabling the construction of asymptotic solutions to arbitrary accuracy with respect to a~small parameter. The study also addresses the approximate description of two-phase soliton- and peakon-like solutions.

Theorems on the asymptotic accuracy of the constructed asymptotic solutions are also proved. Each of the considered cases is illustrated by nontrivial examples for which,
in accordance with the obtained general results, approximate solutions are derived in explicit form and their graphs are presented.
\end{abstract}


\maketitle


\vskip5mm


\section{Introduction}

Among nonlinear systems, the Camassa--Holm (CH) equation has received considerable attention due to its hydrodynamical significance and rich mathematical structure. It was derived as a new completely integrable dispersive shallow-water equation in \cite{CamassaHolm}, where the existence of a new type of soliton-like solution, called peakons, was discovered. These solutions, similarly to classical solitons, describe localized waves that interact without losing their identities after collision; however, unlike smooth solitons, peakons possess a sharp peak at their crest. The CH equation is written as
\begin{equation} \label{Camassa_Holm}
u_t + 2\kappa u_x - u_{txx} = - 3 uu_x + 2 u_x u_{xx} + uu_{xxx},
\end{equation}
where $ u $ denotes the height of the water's free surface above a flat bottom and $ \kappa $ is a constant related to the critical shallow water speed.
Moreover, it was in this equation that peakons were first discovered as soliton solutions whose limiting form, as $ \kappa \to 0 $ develops peaks \cite{Lundmark_2022, Lundmark_2019}. These peaked solitons (``peakons'') are given by $ u(x, t) = c \exp(-|x-ct|) $ and describe solitary waves traveling with a~speed $ c > 0 $ and having a peak of height $ c $. If $ c < 0 $, then the wave moves to the left with a downward peak and it is called an antipeakon.

It should be noted, that equation (\ref{Camassa_Holm}) was originally obtained as a bi-Hamil\-to\-nian generalization of the Korteweg--de Vries equation by A.~Fokas and B.~Fuchs\-steiner in 1981 through the method of recursion operators \cite{Focas}. In the sequel, the relation between the CH equation and various physical phenomena was discovered. In particular, it arises in the study of second-grade non-Newtonian fluids in the vanishing viscosity limit~\cite{Busuioc}, in the dynamics of nonlinear dispersive waves in a compressible Mooney–Rivlin rod \cite{Dai}, in the modeling of turbulent flows \cite{Chen_1}, and plays a significant role in examining wave breaking \cite{Constantin-Escher, Constantin_Lannes, BrandoleseCMP, McKean_2004}, and others. These and other applications have stimulated extensive research devoted to the analysis of this equation.

Various problems for the CH equation, including the existence of solitons and other types of particular solutions to equation \eqref{Camassa_Holm}, have been intensively studied in numerous papers using a wide range of approaches and methods, beginning with \cite{Alber_2000, Cooper, Olver, Danchin_2003}. A significant contribution to the study of the CH equation is due to A.~Constantin, particularly concerning the existence, uniqueness, local well-posedness in Sobolev spaces, stability, and regularity of its solutions \cite{Constantin_1997, Constantin_Strauss_2000, Constantin_2000, Constantin_2001, Constantin_2006, Constantin_2007, Constantin_2007_1}. In these papers he and his coauthors applied classical PDE methods such as energy estimates, Sobolev space techniques and inverse scattering transform (IST) approach. As in the case of other integrable systems, the IST has become a powerful tool for studying the CH equation, allowing the construction of its reflectionless potentials as well as exact solutions of different types \cite{Johnson, Li}. In this connection, it is worth mentioning the Riemann--Hilbert approach, which additionally enables the study of the qualitative behavior of solutions under various initial and boundary conditions, including step-like Cauchy data \cite{Minakov_2016, Shepelskiy, Shepelsky_2009}.

Another powerful tool for analyzing equation \eqref{Camassa_Holm}, particularly for constructing quasi-periodic or finite-gap solutions, are the algebro-geometric and isospectral methods \cite{Alber_2000, Schiff_1996, Schiff_1998, Teschl_2013}. These techniques yield a broad class of algebro-geometric solutions of the CH equation~\cite{Qiao}, including solitons as particular solutions.
Similarly to the one-peakon solution, which is obtained from the smooth soliton via the limiting procedure as $ \kappa \to 0 $, this method was also successfully applied to derive the multipeakon solution using the $PQ$-decomposition technique \cite{Parker}.

The multipeakon solutions can be written as \cite{Lundmark_2022, Holden_2007, Anco}
\begin{equation} \label{multi-peakon}
u(x, t) = \sum\limits_{k=1}^N m_k(t) e^{-|x-x_k(t)|},
\end{equation}
where the function $ x_k(t) $ describes the position of the $k$-th peakon, while $ m_k(t) $ denotes its amplitude, and for the peakons with positive amplitudes this relation exhibits nonlinear interactions similar to those of the classical \linebreak smooth solitons. In the other cases, when the amplitudes of the peakons have different signs, interactions among them can occur in various scenarios. Even in the simplest case of a collision between a single peakon with amplitude $ m_1(t) > 0 $ and a~single antipeakon with amplitude $ m_2(t) < 0 $ the behavior depends on the type of solution considered \cite{Lundmark_2022}.

In the {\it conservative case}, where the energy integral
\[
 E(t) = \int_{{\mathbb R}} \bigl( u^2 + u_x^2\bigr) dx
\]
remains constant up and after collision, the solution describes the continued propagation of a peakon and an anti\-pea\-kon after the collision, with their amplitudes changing signs, so that the new amplitudes satisfy the inequality $ m_1(t) < 0 < m_2(t) $. A peakon–anti\-peakon collision may also follow another scenario, namely the {\it dissipative solution}, in which the total energy decreases at the collision time, and the solution continues as a single peakon traveling with a new velocity. There exists also one more possibility, known as the $ \alpha $-{\it dissipative solution}, characterising the reappearance of the peakon and antipeakon after the collision, with only a fraction $ 1 - \alpha $ of the lost energy restored, where $ 0 < \alpha < 1 $. In~\cite{Artebrant}, the authors demonstrated that a certain plateau-shaped traveling wave solution appears when two-peakon solutions interact, and how this plateau subsequently breaks up.
A detailed analysis of the complex dynamics of the peaked solitons associated with the CH equation can be found in numerous studies (see, for example, \cite{Lundmark_2022,Parker_2008, Coclite, Feng}).

It should be noted that peakon solutions and their remarkable properties represent only one aspect of the intricate structure of the CH equation.
Similarly to the existence of two wave-breaking scenarios for the KdV equation -- the rough scenario in which the solution becomes unbounded and the milder scenario in which the function remains continuous while its derivative becomes unbounded -- a similar phenomenon also occurs for the CH equation. Particularly, this equation admits so-called cuspon solutions \cite{Lenells_2005} whose derivative grows without bound, as well as peakon solutions corresponding to a milder scenario in which the derivative exhibits a jump discontinuity but remains finite. Analogously to peakon and soliton solutions, cuspons exhibit the same interaction property: they preserve their shapes and velocities after interaction \cite{Lenells}. The properties of cuspon solutions have been studied less extensively due to their specific differentiability features, which complicate the application of standard analytical methods and require modifications of existing techniques.

In the sequel, other equations admitting peakon and cuspon solutions were discovered, among which are
the Degasperis--Procesi (DP) equation \cite{Degasperis_Procesi}
\begin{equation}\label{Degasperis_Procesi}
u_t - u_{txx} = - 4 uu_x + 3 u_x u_{xx} + uu_{xxx},
\end{equation}
the Novikov equation \cite{Novikov}
\begin{equation}\label{Novikov}
u_t - u_{txx} = - 4 u^2 u_x + 3 u u_x u_{xx} + u^2u_{xxx},
\end{equation}
the Focas--Fuchssteiner--Olver--Poseanu--Qiao (FORQ) equ\-a\-tion \cite{Focas, Olver}
\begin{equation}\label{FORQ}
u_t - u_{txx} = \bigl(- u^3 + u u^2_x + u^2 u_{xx} - u_x^2u_{xx}\bigr)_x.
\end{equation}

Another important aspect of soliton theory is the study of integrable systems in the small-dispersion regime, which is characterized by the presence of a small parameter at the highest derivatives. In this context, the KdV equation was the first for which the zero-dispersion limit was analyzed \cite{Lax1, Lax2, Lax3, Venakides_1990} using the inverse scattering transform (IST) and WKB technique \cite{Miura}.

Extensions of such studies concern the other systems, including the CH equation, as discussed in~\cite{Grava_2009}:
\begin{equation}\label{CH_Grava}
u_t + (3 u + 2 \nu ) u_x - \varepsilon^2 (u_{xxt} + 2 u_x u_{xx} + uu_{xxx} ) = 0,
\end{equation}
where $ \nu $ is a constant related to the critical shallow water speed, and $ \varepsilon $ is a~constant proportional to the mean water depth. Using numerical methods, the authors performed a detailed analysis and identified Whitham zones of oscillations in the solution of the Cauchy problem for equation \eqref{CH_Grava} with the initial data $ u_0(x) = -\sinh^2 (x) $.

Owing to the broad variety of solutions of the CH equation, it is natural to extend equation \eqref{Camassa_Holm} to the case $\kappa = 0$ by incorporating variable coefficients into the model \cite{Vaneeva_2017,Popovych_2010}. This modification is motivated by the need to study wave processes in media with spatially and temporally varying properties. Consequently, this paper addresses the CH equation with variable coefficients (vcCH equation) and a singular perturbation, characterized by a small dispersion:
\begin{equation}\label{CHolm_vc}
a(x, t, \varepsilon) u_t + b(x, t, \varepsilon) uu_x - \varepsilon^2 (u_{txx} + 2 u_x u_{xx} + uu_{xxx} ) = 0.
\end{equation}

We assume that the functions
$ a(x, t, \varepsilon) $ and $ b(x, t, \varepsilon) $ with $ (x, t) \in {\mathbb R} \times [0; T] $ for some $ T > 0 $ can be written as asymptotic series
\begin{equation} \label{coeff}
	a(x, t, \varepsilon) = \sum\limits_{k=0}^\infty \varepsilon^k a_k(x, t), \, \, b(x, t, \varepsilon) = \sum\limits_{k=0}^\infty \varepsilon^k b_k(x, t),
\end{equation}
and $ a_0(x, t) b_0(x, t)\not= 0 $ for all $ (x, t) \in {\mathbb R} \times [0; T] $.

In general, solutions of the vcCH system \eqref{CHolm_vc} cannot be obtained in analytical form due to the presence of variable coefficients. Therefore, we focus on constructing its asymptotic solutions that exhibit soliton-like features, including peakon solutions. These solutions are represented as a sum of a regular part, which serves as a background function common to both soliton-like and peakon-like solutions, and singular components, which capture the distinctive features of these solutions. In the following sections, we present a detailed description of how these components are determined.

The presented algorithm for constructing the approximate solutions is based on the general methodology for deriving asymptotic soliton-like solutions of hydrodynamic type systems with variable coefficients and a small dispersion, with an approach specifically adapted to equation \eqref{CHolm_vc}. Such an adaptation is necessary because each nonlinear wave-like equation requires a specific treatment, as has been demonstrated for the vcKdV equation \cite{Sam_2005}, the vcBBM equation \cite{Sam_JMP}, the vcBurgers equation \cite{Sam_Burgers} and the vcmCH equation \cite{SamBrandSam}.

The papers is organized as follows.
We construct one-phase (Section~\ref{section2}) and two-phase soliton-like solutions (Section~\ref{section3}), as well as one-phase (Section~\ref{section4}) and two-phase peakon-like solutions (Section~\ref{section5}) of the vcCH equation~\eqref{CHolm_vc}. We also evaluate the accuracy of the constructed solutions and demonstrate the effectiveness of the proposed algorithm through examples in each case, by deriving exact solutions and presenting graphs of their first asymptotic approximation.

\section{One-phase soliton-like solutions}\label{section2}

Our analysis begins with the construction of asymptotic soliton-like solutions to equation~\eqref{CHolm_vc}, covering both one-phase and two-phase asymptotic structures. Following the idea of the nonlinear WKB method \cite{Miura}, we propose an algorithm for constructing an approximate (asymptotic) solution in the form of an expansion with respect to a small parameter. We describe a procedure that makes it possible to determine the terms of the expansion to an arbitrary order of smallness, although it should be noted that, in practice, one usually restricts its attention to the main term of asymptotics and, possibly, to its the first approximation.

Firstly, we specify the form of the asymptotic solutions. Accordingly, the general methodology of asymptotic soliton- and peakon-like solutions to singularly perturbed partial differential equations with variable coefficients these functions are searched as \cite{SamBrandSam}
\begin{equation}\label{as_sol}
u(x, t, \varepsilon) = \sum\limits_{j=0}^{\infty} \varepsilon^j [ u_j(x, t) + V_j (x, t, \tau) ] , \, \, \tau = \frac{x - \varphi(t)}{\varepsilon},
\end{equation}
where $ u_j(x, t) $ and $ V_j(x, t, \tau) $, $ j = 0,1, \dots $, are infinitely differentiable functions. The function $ \varphi(t) $ is referred to as the phase function and is assumed to be smooth. In the case of constant coefficients equations mentioned above, $ \varphi(t) $ is linear; however, for variable-coefficients systems, it undergos deformation and has to satisfy additional conditions.

The functions $ u_j(x, t) $, $ j = 0,1, \dots $, form the regular part of asymptotics~\eqref{as_sol}; they represent the background of the wave propagation. In contrast, the functions $ V_j (x, t, \tau) $, $ j = 0,1, \dots $, constitute the singular part of the asymptotics and are introduced to reproduce the soliton characteristics of the asymptotic solution.
This leads us to impose appropriate functional constraints to the singular terms $ V_j (x, t, \tau)$.
To this end, we introduce the functional spaces that will be used.

\subsection{Main definitions}

By $ {\overline C}^\infty_0({\mathbb R}) $ we denote the space of infinitely differentiable functions $ u(x) $, $ x \in {\mathbb R} $, satisfying the relation
$$
\frac{d^{ n} u(x)}{dx^n} \to 0
$$
as $ |x| \to +\infty $ for any non-negative integer $ n $.

Let $ S = S({\mathbb R}) $ be a Schwartz space, i.e., the space of infinitely differentiable rapidly decreasing on $ {\mathbb R} $ functions such that for any integers $ m, n \ge 0 $ the condition
$$
\sup\limits_{x \in {\mathbb R} } \bigl| {x^m D^n u (x)} \bigr| < + \infty
$$
holds.

By $\mathrm{H}_s({\mathbb R}) $, $ s \in {\mathbb R} $, we denote the Sobolev space \cite{Fritz, Evans}, whose elements belong to the space of tempered distributions $\mathcal{S}^*({\mathbb R})$, and their Fourier transforms $ F[g](\xi) $ satisfy the condition
\begin{equation}\label{norma}
	|| g ||_s^2 = \int_{-\infty}^{+\infty} \bigl(1 + |\xi|^2\bigr)^s | F[g](\xi) |^2 d \xi < + \infty .
\end{equation}

In constructing the singular terms of the asymptotic one-phase soliton-like solutions, we used the functional spaces $G_1$ and $G_0$ whose elements are chosen to reflect the characteristic properties of solitons \cite{Sam_2005, Maslov_book}. By $ G_1 $ we denote the space of infinitely differentiable functions
$f\colon {\mathbb R} \times [0;T] \times {\mathbb R} \to {\mathbb R} $
satisfying the two following conditions:
\begin{enumerate}
\item[1)]
For any non-negative integers $ n $, $ p $, $ q $ and $ r $
$$
\lim_{\tau \to + \infty} \tau^n \frac{\partial ^p}{\partial x^p} \frac{\partial ^q}{\partial t^q} \frac{\partial ^r}{\partial \tau^r} f (x, t, \tau) = 0,
$$
uniformly with respect to $ (x, t) \in K $, in any compact set $ K\subset {\mathbb R} \times [0;T] $.
\item[2)]
There exists a differentiable function
$$
f^-\colon {\mathbb R} \times[0;T]\to {\mathbb R}
$$
such that for any non-negative integers $ n $, $ p $, $ q $ and $ r $
$$
\lim_{\tau \to - \infty} \tau^{n} \frac{\partial ^p}{\partial x^p} \frac{\partial ^{q}}{\partial t^{q}}
 \frac{\partial ^{r}}{\partial \tau^{r}} \left( f (x, t, \tau) - f^{-}(x, t)\right) = 0,
$$
uniformly with respect to $ (x, t) \in K $, in any compact set $ K\subset {\mathbb R} \times [0;T] $.
\end{enumerate}

By $ G_0$ we denote the subspace of $G_1$, consisting of functions $ f\in G_1 $ such that
$$
\lim_{\tau \to - \infty} f (x, t, \tau) = 0
$$
uniformly with respect to the variables $ (x, t)\in K $, in any compact set $ K \subset {\mathbb R} \times [0;T] $.

We denote by $ {\widetilde G}_1 $ the space of infinitely differentiable functions $ g\colon [0; T] \times {\mathbb R} \to {\mathbb R} $, satisfying the two following conditions:
\begin{enumerate}
\item[1)]
For any non-negative integers $ n $, $ p $ and $ q $
$$
\lim_{\tau \to + \infty} \tau^n \frac{\partial ^p}{\partial t^p} \frac{\partial ^q}{\partial
 \tau^q} g (t, \tau) = 0,
$$
uniformly with respect to $ t \in [0; T] $.
\item[2)]
There exists a differentiable function $ g^-\colon [0;T]\to {\mathbb R} $ such that for any non-negative integers $ n $, $ p $ and $ q $
$$
\lim_{\tau \to - \infty} \tau^{n} \frac{\partial ^p}{\partial t^p} \frac{\partial ^{q}}{\partial \tau^{q}}
 ( g (t, \tau) - g^{-}(t) ) = 0,
$$
uniformly with respect to $ t \in [0; T] $.
\end{enumerate}

Let $ {\widetilde G}_0 $ be the subspace of $ {\widetilde G}_1 $, consisting of functions $ g\colon [0; T] \times {\mathbb R} \to {\mathbb R} $, such that
\[
\lim_{\tau \to - \infty} g (t, \tau) = 0,
\]
uniformly with respect to $ t \in [0; T] $.

We adopt the following definition, which plays a crucial role in this paper.

\begin{definition}[\cite{Sam_2005, Maslov_book}] \label{def_1}
A nontrivial function $$ u = u(x, t, \varepsilon),\quad (x, t) \in {\mathbb R} \times [0;T], $$ where $ \varepsilon>0 $ is a small parameter, is called an asymptotic one-phase so\-li\-ton-like function if for any integer $ N \ge 0 $ it can be repre\-sen\-ted as
\begin{align}
& u(x, t, \varepsilon) = \sum_{j=0}^N \varepsilon^j
		\left[u_j(x, t) + V_j \left(x, t, \tau \right)\right] + O(\varepsilon^{N+1}), \nonumber\\
& \tau = \frac{x - \varphi(t)}{\varepsilon}, \label{soliton_one-phase}
\end{align}
where $ \varphi\in C^{\infty} ([0;T]) $ is a scalar function, $u_j\in C^\infty ({\mathbb R}\times [0;T]) $ for $ j = 0, 1, \dots , N$, nontrivial function $V_0\in G_0, $ and $ V_j \in G_1 $ for $ j = 1, \dots, N $.
\end{definition}

Here and in the sequel, we use the notation of asymptotic analysis of the following form $ \Psi (x,t, \varepsilon ) = O\bigl( \varepsilon^N
\bigr) $, where $ \varepsilon > 0 $ is a small parameter. It means that there exist such values $ \varepsilon_0 > 0 $ and $ C > 0 $ that the inequality $ | \Psi (x,t,
\varepsilon ) | \le C \varepsilon^N $ holds for all $ \varepsilon \in (0; \varepsilon_0)$, $ (x,t) \in K $, where $ K \subset {\mathbb R} \times [0;T] $ is a compact set and value $ C $ is only depending on the number $ N $ and the set $ K $.

Move on to a problem of constructing the asymptotic soliton-like solutions to the vcCH equation~\eqref{CHolm_vc}.

\subsection{Differential equations for the asymptotical terms}\label{section2.2}

The asymptotic one-phase soliton-like solutions are \linebreak searched in the form \eqref{soliton_one-phase}. By substituting this expansion into equation \eqref{CHolm_vc} and limiting as $ \tau \to + \infty $, due to the property of the singular terms we deduce for the regular part of the asymptotics
$$
U_N(x, t, \varepsilon) = \sum_{j=0}^N \varepsilon^j u_j(x, t), \quad N \ge 0,
$$
the following partial differential equations
\begin{align} \label{reg_part_0}
& a_0(x, t) \frac{\partial u_0}{\partial t} + b_0(x, t) u_0 \frac{\partial u_0}{\partial x} = 0,
\\
 \label{reg_part_j}
& a_0(x, t) \frac{\partial u_j}{\partial t} +
b_0(x, t) \frac{\partial}{\partial x} (u_0 u_j ) = f_j(x, t), \, j = 1, 2, \dots, N,
\end{align}
where the function $ f_j(x, t) = f_j(x, t, u_0, u_1, \dots, u_{j-1}) $
is determined in a recursive manner after the corresponding functions $ u_0 , u_1 , \dots , u_{j-1}$ have been constructed.

In particular,
\begin{align*}
f_1(x, t) ={}& - a_1(x, t) \frac{\partial u_0}{\partial t} - b_1(x, t) u_0 \frac{\partial u_0}{\partial x}, \\
f_2(x, t) ={}& - a_1(x, t) \frac{\partial u_1}{\partial t} - a_2(x, t) \frac{\partial u_0}{\partial t} - b_0(x, t) u_1 \frac{\partial u_1}{\partial x} \\
& - b_1(x, t) \frac{\partial}{\partial x} \left( u_0 u_1 \right) - b_2(x, t) u_0 \frac{\partial u_0}{\partial x}  \\
& +\frac{\partial^3 u_0}{\partial t\partial x^2} + 2 \frac{\partial u_0}{\partial x} \frac{\partial^2 u_0}{\partial x^2} + u_0 \frac{\partial^3 u_0}{\partial x^3}.
\end{align*}

The regular terms of the asymptotics may be considered known, as the quasi-linear equation \eqref{reg_part_0} and linear equation \eqref{reg_part_j} are solvable by means of the method of characteristics.

Let us proceed to consideration of the singular part of the asymptotics
$$
V_N(x, t, \tau, \varepsilon) = \sum_{j=0}^N \varepsilon^j V_j(x, t, \tau), \quad N \ge 0.
$$
Its terms satisfy the following third order partial differential equations
\begin{align}
& - a_0(x, t) \varphi' \frac{\partial V_0}{\partial \tau} + b_0(x, t) ( u_0(x, t) + V_0 ) \frac{\partial V_0}{\partial \tau} \nonumber\\
&\qquad {}  + \varphi' \frac{\partial^3 V_0}
	{\partial\tau^3} - 2 \frac{\partial V_0}{\partial\tau} \frac{\partial^2 V_0}{\partial\tau^2} - V_0 \frac{\partial^3 V_0}{\partial\tau^3} = 0,\label{singular_part_0_soliton_sol}
\\
&
- a_0(x, t) \varphi' \frac{\partial V_j}{\partial \tau} + b_0(x, t) u_0(x, t) \frac{\partial V_j}{\partial \tau} + \varphi' \frac{\partial^3 V_j} {\partial\tau^3} \nonumber\\
& \qquad {}+ b_0(x, t) \frac{\partial}{\partial\tau} (V_0 V_j ) - 2 \frac{\partial}{\partial\tau} \left(\frac{\partial V_0}{\partial \tau} \frac{\partial V_j}{\partial \tau} \right)
\nonumber\\
& \qquad {} - V_0 \frac{\partial^3 V_j} {\partial\tau^3} - \frac{\partial^3 V_0}{\partial\tau^3} V_j = {F}_j (x, t, \tau), \quad j= 1, \dots, N,\label{singular_part_1_soliton_sol}
\end{align}
where the function
$$ {F}_j (x,t, \tau) = F_j(t, V_0(x, t, \tau), \dots , V_{j-1} (x, t, \tau))
$$
is defined recursively after the corresponding functions $ V_0 $, $ V_1 $, $ \dots $, $ V_{j-1} $ have been determined.

The analysis of the system \eqref{singular_part_0_soliton_sol}, \eqref{singular_part_1_soliton_sol} constitutes one of the central aspects in the construction of asymp\-to\-tic soliton-like solutions to the vcCH equation.

Although equation \eqref{singular_part_0_soliton_sol} is quasi-linear and homogeneous, and equation \eqref{singular_part_1_soliton_sol} is linear and non-homogeneous, their solvability must be studied within the function spaces $ G_0 $, $ G_1 $ introduced above, which represents a significant complication.
So, taking into account the properties of the functions $ V_j(x, t, \tau) $, $j= 1, \dots, N $, we may proceed to study the system \eqref{singular_part_0_soliton_sol}, \eqref{singular_part_1_soliton_sol} as follows. Firstly, we assume the function $ \varphi = \varphi(t) $ is known, and equations \eqref{singular_part_0_soliton_sol}, \eqref{singular_part_1_soliton_sol} are considered on the discontinuity curve
$$
\Gamma = \{(x, t)\in {\mathbb R}\times[0; T] \mid x - \varphi(t) =0 \},
$$
and $ t $ is considered below as a parameter.

Secondly, the functions determined along the curve $ \Gamma $ are extended to the desired domain.

We turn our attention to a detailed discussion of the algorithm. By notation
\begin{equation}\label{sing_term_curve}
v_j = v_j (t, \tau) = V_j (x, t, \tau) \bigr|_{\Gamma}, \quad j = 0, 1, \dots, N,
\end{equation}
from \eqref{CHolm_vc}, \eqref{singular_part_0_soliton_sol}, \eqref{singular_part_1_soliton_sol} we deduce differential equations for the functions $ v_j (t, \tau) $:
\begin{align}
 & ( \varphi' - v_0 ) \frac{\partial^3 v_0}{\partial\tau^3} + ( b_0 u_0 - a_0 \varphi' ) \frac{\partial v_0}{\partial \tau} + b_0 v_0 \frac{\partial v_0} {\partial\tau} \nonumber \\
 & \qquad{} - 2 \frac{\partial v_0}{\partial\tau} \frac{\partial^2 v_0}{\partial\tau^2}
	= 0, \label{sing_part_02_soliton_sol_1}
\\
& ( \varphi' - v_0 ) \frac{\partial^3 v_j} {\partial\tau^3} + ( b_0 u_0 - a_0 \varphi' ) \frac{\partial v_j}{\partial \tau} +
 b_0 \frac{\partial }{\partial \tau } ( v_0 v_j )
\nonumber\\
& \qquad{} - 2 \frac{\partial }{\partial \tau } \left( \frac{\partial v_0}{\partial \tau} \frac{\partial v_j}{\partial \tau} \right)	-
	\frac{\partial^3 v_0}{\partial \tau^3} v_j = {\cal F}_j, \quad j = 1, \dots, N ,\label{sing_part_j2_soliton_sol_1}
\end{align}
where the function
$$ {\cal F}_j = {\cal F}_j (t, \tau) = F_j(t, V_0(x, t, \tau), \dots , V_{j-1} (x, t, \tau)) \bigr|_{ x 	= \varphi(t)}
$$ 
is found recursively after the functions $ V_0 , V_1 , \dots , V_{j-1} $ have been calculated and their restrictions to the curve $ \Gamma $ have been evaluated.
In particular,
\begin{align}
{\cal F}_1 (t, \tau) ={}& \frac{\partial^3 v_0}{\partial\tau^2 \partial t} - a_0 \frac{\partial v_0}{\partial t} +
 ( \tau a_{0x} + a_1 ) \varphi' \frac{\partial v_0}{\partial \tau}
\nonumber\\
& - \left[ \tau \frac{\partial}{\partial x} \bigl( b_{0} u_0 \bigr) + b_0 u_1 + b_1 u_0 \right] \frac{\partial v_0}{\partial \tau}
+ (\tau b_{0x} + b_{1} ) v_0 \frac{\partial v_0}{\partial \tau}.\label{function_F}
\end{align}
Here and below $ a_k = a_k(\varphi, t)$, $ b_k = b_k(\varphi, t)$, $ u_k = u_k(\varphi, t)$, $ k = 0, 1 $, $ \varphi = \varphi(t) $, $ t \in [0; T]$.

\subsection{The main singular term of the soliton-like solution}

Determining the main term of the singular part of the asymptotics is a~crucial step in the approximate description of modulated soliton solutions, that arise due to the system’s emerging variable characteristics. In contrast to the analogous problem for the vcKdV  \cite{Sam_2005}, vcBBM \cite{Sam_JMP}, and vcmCH \cite{SamBrandSam} equations, in present case the main term can only be expressed in implicit form. 
The other terms of the singular part can be found in exact form.

We now proceed to the analysis of equation~\eqref{sing_part_02_soliton_sol_1}. Upon integration, it yields the following second-order differential equation:
\begin{multline} \label{sing_part_02_soliton_sol_2}
( \varphi' - v_0 ) \frac{\partial^2 v_0}{\partial\tau^2} + ( b_0 u_0 - a_0 \varphi' ) v_0  +\frac{1}{2} b_0 v_0^2 - \frac{1}{2} \left(\frac{\partial v_0}{\partial\tau}\right)^2 = 0,
\end{multline}
which can be equivalently rewritten as a system of first-order differential equations
\begin{equation} \label{sing_part_02_soliton_sol_3}
	\begin{aligned}
 \frac{d v_0}{d \tau} & = y, \\
	\displaystyle \frac{d y}{d \tau}& = \frac{y^2 - b_0 v_0^2 + 2 (a_0 \varphi' - b_0 u_0) v_0}{2(\varphi' - v_0)}.
	\end{aligned}
\end{equation}

The first integral of the system \eqref{sing_part_02_soliton_sol_3} is a function
$$
H(v_0, y) = b_0 v_0^3 - 3 (v_0 - \varphi')y^2 - 3 ( a_0 \varphi' - b_0 u_0 ) v_0^2.
$$

Accordingly properties of the function $ v_0 $ at the infinity in $ \tau $, we find the first-order differential equation of the form
\begin{equation} \label{v_0}
	\left(\frac{d v_0}{d \tau} \right)^2 = \frac{b_0 v_0^3 - 3 ( a_0 \varphi' - b_0 u_0)v_0^2}{3 (v_0 - \varphi')}.
\end{equation}

We focus on particular solutions of equation \eqref{v_0} which possess a soliton character.
Analogously to the Camas\-sa–Holm equation \cite{Li, Schiff_1998}, such a solution can be represented in implicit form as
\begin{equation} \label{main_term}
v_0(t, \theta) = \frac{ 3 ( a_0 \varphi' - b_0 u_0 ) }{b_0} \, \left(\cosh^2 \theta - \frac{3 ( a_0 \varphi'- b_0 u_0)}{b_0 \varphi'} \sinh^2 \theta\right)^{-1}.	
\end{equation}
Here, the variable $ \theta $ is defined in terms $ \tau $ through the relation
\begin{equation}\label{theta}
\tau = 2 \sqrt{ \frac{ \varphi'}{a_0 \varphi' - b_0 u_0}} \theta + \sqrt{\frac{3}{b_0}}\ln{\frac{\cosh(\theta - \theta_0)}{\cosh(\theta + \theta_0)}},
\end{equation}
where
\[
\theta_0 = \theta_0 (t) = \tanh^{-1} \sqrt{\frac{3(a_0 \varphi' - b_0 u_0)}{b_0 \varphi'}}.
\]
The last equalities require the following assumptions:
\begin{equation}\label{cond_w}
	 0 < \frac{a_0 \varphi' - b_0 u_0}{b_0 \varphi' } < \frac{1}{3}, \quad b_0 > 0,
\end{equation}
which specify the interval $ [0; T] $ on which the asymptotics is constructed.

Let us note that, in the case when the main term of the regular part of the asymptotic expansion is zero, i.e., $ u_0(x, t) = 0 $, condition \eqref{cond_w} reduces to the following simple inequality:
\begin{equation*}
0 < \frac{a_0}{b_0} < \frac{1}{3}.	
\end{equation*}

The function $v_0(t, \theta)$ in \eqref{main_term} decays rapidly with respect to $\theta$ and, according to Lemma~\ref{lem: lemma_1}, also with respect to $\tau$, which is important for our analysis.

\begin{lemma}\label{lem: lemma_1}
$$
\tau \sim 2 \sqrt{ \frac{ \varphi^{ \prime}}{a_0 \varphi^{ \prime} - b_0 u_0}} \theta \quad \mbox {as} \quad |\theta| \to \infty.
$$	
\end{lemma}

\begin{proof}
To prove the lemma, it suffices to observe that the function
$$
f(\theta) = 2 \sqrt{ \frac{ \varphi'}{a_0 \varphi' - b_0 u_0}} \theta + \sqrt{\frac{3}{b_0}}\ln{\frac{\cosh(\theta - \theta_0)}{\cosh(\theta + \theta_0)}}
$$
is strictly increasing on $ \theta \in {\mathbb R} $ and the mapping $ f: {\mathbb R} \mapsto {\mathbb R} $ is bijective.
\end{proof}

Therefore, $v_0(t, \tau) \in {\widetilde G}_0$, and it can be extended from $\Gamma$ to $V_0(x, t, \tau)$ as $V_0(x, t, \tau) = v_0(t, \tau)$, and the resulting function belongs to the space $G_0$. Consequently, the main term of the asymptotic one-phase soliton-like solution is written as
\begin{equation}\label{Y_0}
Y_0(x, t, \varepsilon) = u_0(x,t) + v_0(t, \tau).
\end{equation}

\subsection{Higher singular terms}

The starting point for defining the higher-order singular terms is equation~\eqref{sing_part_j2_soliton_sol_1}, where we consider $ j = 1, \dots, N $. The procedure for finding the singular terms is carried out in two steps: first, we establish the conditions for the existence of solutions to equation~\eqref{sing_part_j2_soliton_sol_1} in the space $ {\widetilde G}_1 $ and determine their representation form; then, we extend the obtained functions from the curve~$\Gamma$.

We begin by integrating \eqref{sing_part_j2_soliton_sol_1} in $ \tau $ from $ - \infty $ to $ \tau $:
\begin{equation}\label{sing_part_j_1}
( \varphi' - v_0) v_j'' - v_j' v_0' + [ b_0 ( v_0 + u_0 ) - a_0 \varphi' - v_0''] v_j
 = \Phi_j (t, \tau),
\end{equation}
where
\begin{equation}\label{function_Phi}
\Phi_j (t, \tau) = \int_{-\infty}^{\tau} {\cal F}_j(t, \xi) d\xi + E_j(t), 
\end{equation}
and the constant of integration $ E_j(t) $ does not depend on the variable $ \tau $. It can be chosen from formula \eqref{function_Phi} under the condition $ \lim_{\tau \to +\infty} \Phi_j (t, \tau) = 0 $.

According to the definition of asymptotic soliton-like solutions, we need to study the conditions for the existence of solutions to equation \eqref{sing_part_j_1} in the space~$\widetilde{G}_1$. To this end, we introduce the differential operator
\begin{align} \label{operator_L}
	L = L\left(t, \tau, \frac{d}{d \tau}\right): = \rho \frac{d^2}{d\tau^2} - v_{0\tau} \frac{d}{d\tau} + b_0 ( v_0 + u_0 ) \nonumber \\
 - a_0 \varphi' - v_{0\tau\tau},
\end{align}
where $ \rho = \rho(t, \tau) = \varphi'(t) - v_0 (t, \tau) $, $ t \in [0;T] $,
and rewrite linear differential equation \eqref{sing_part_j_1} in the operator form
\begin{equation}\label{equation_L}
	L v = \Phi.
\end{equation}

Recall that $ v = v(t, \tau) $, $ \Phi = \Phi(t, \tau) $, and $ t \in [0;T] $ is a parameter.

We use operator equation \eqref{equation_L} to find conditions under which differential equation \eqref{sing_part_j_1} has a solution in the space $ {\widetilde G}_1 $. To do it, we apply the results of the theory pseu\-do\-differential operators \cite{Hor}, in particular from
\cite{GR1, GR2}.

\subsubsection{Solvability of operator equation \eqref{equation_L} in the space $ \mathcal{S}({\mathbb R}) $ }

In the sequel, we use the result, which will be essential for our analysis.

\begin{theorem} \label{thm: th_1}
Let the following assumptions hold for all $t \in [0, T]$:
\begin{enumerate}
\item[1.]
Inequalities \eqref{cond_w} take place;
\item[2.]
The function $ \tau\mapsto \Phi(t,\tau)$ belongs to $\mathcal{S}(\mathbb R) $.
\end{enumerate}

Then, for any $t\in[0; T]$ equation \eqref{equation_L} has a solution $v(t,\cdot) $ in the space $ \mathcal{S}(\mathbb R) $ if and only if the function $ \Phi $
satisfies the orthogonality condition of the form
\begin{equation} \label{ort_cond}
\int_{-\infty}^{+\infty} \Phi_\tau (t, \tau) v_0(t, \tau) d \tau = 0, 
\end{equation}
where the function $ v_0 $ is given by formula \eqref{main_term} and $ t\in [0;T] $.
\end{theorem}

\begin{proof}
For proving the theorem it is enough to take into account that the symbol of the operator $ L $ can be written as
$$
p_L(t, \tau, \xi) = - \rho \xi^2 + i \xi v_{0\tau} + b_0 (v_0 + u_0) - a_0 \varphi' - v_{0\tau\tau} .
$$

Since $ v_0(t, \tau) \in {\mathcal S} ({\mathbb R}) $, it can be shown, analogously to~\cite{SamBrandSam}, that for any $ s \in {\mathbb R} $ the operator $ L\colon H_{s+2} \to H_s $ is Noetherian.

The operator $ L^* $, adjoint to operator \eqref{operator_L}, can be written as
$$
L^* = \frac{d^2}{d\tau^2} \rho + \frac{d}{d\tau} v_{0\tau} + b_0 (v_0 + u_0) - a_0\varphi' - v_{0\tau\tau}.
$$
In virtue of equation \eqref{sing_part_02_soliton_sol_1}, the function $ v_{0\tau} $ belongs to the kernel of the operator $ L^*\colon S ({\mathbb R}) \to S ({\mathbb R}). $
Another linear independent solution to the equation $ L^* v =0 $ is given as
$$
w_0 (t, \tau) = v_{0\tau} (t, \tau) \int_{\tau_0}^\tau \frac{d \xi}{\rho(t, \xi) v_{0 \xi }^2 (t, \xi)}, 
$$
where $ \tau_0 \in [ - \infty ; + \infty ) $.

Considering the Wronskian of the solutions $ v_{0 \tau} (t, \tau)$ and $ w_0 (t, \tau)$, and using its constancy, we let $ \tau \to +\infty $ to conclude that $ w_0 \not\in S({\mathbb R}) $. This observation allows us to formulate the orthogonality condition -- which provides the necessary and sufficient condition for the existence of a solution to equation \eqref{equation_L} in the space $ S ({\mathbb R}) $ -- in the form
\begin{equation}\label{orthog_cond_2}
\int_{-\infty}^{+\infty} \Phi(t, \tau) v_{0\tau} (t, \tau) d \tau =0, \quad t \in[0;T].
\end{equation}
	
Summarizing the arguments above, we conclude that equation \eqref{equation_L} has a~solution in the space $ \mathcal{S}({\mathbb R}) $ if and only if the orthogonality condition \eqref{orthog_cond_2} is satisfied. Due to the property $ \Phi \in \mathcal{S}({\mathbb R}) $, this condition finally leads to \eqref{ort_cond}.

Theorem~\ref{thm: th_1} is proved.
\end{proof}


\subsubsection{Solvability of differential equation \eqref{sing_part_j2_soliton_sol_1} in the space $ {\widetilde G}_1 $ }

Theorem~\ref{thm: th_1} allows to obtain condition for existence of a solution of differential equation \eqref{sing_part_j2_soliton_sol_1} in the space $ {\widetilde G}_1 $.
The following statements are valid.

\begin{lemma}\label{lem: lemma_2}
Let inequality \eqref{cond_w} be true and $ {\cal F}_j \in {\widetilde G_{0}} $. Then equation \eqref{sing_part_j2_soliton_sol_1} has a solution $ v_j \in {\widetilde G}_1 $ if and only if the function $ {\cal F}_j $ satisfies the orthogonality condition of the form
\begin{equation} \label{ort_cond_20}
\int_{-\infty}^{+\infty} {\cal F}_j(t, \tau) v_0(t, \tau) d \tau = 0,
\end{equation}
where the function $ v_0(t, \tau) $ is defined via formula \eqref{main_term} and $ t \in [0;T] $.
\end{lemma}

\begin{proof}
Since the proof of the lemma follows the same reasoning as in \cite{SamBrandSam}, which deals with the construction of asymptotic soliton-like solutions to the vcmCH equation, we provide here only a brief outline of its main steps.

As in the case of the vcmCH equation, we show that the solution $v_j$ of equation~\eqref{sing_part_j2_soliton_sol_1} can be represented as
\begin{equation}\label{vyglyad}
v_j(t, \tau) = \nu_j(t) \eta_j(t, \tau) + \psi_j(t, \tau),
\end{equation}
where
\begin{equation}\label{nu}
\nu_j(t) = - ( a_0 \varphi' - b_0 u_0 )^{-1} \lim_{\tau\to-\infty} \Phi_j(t, \tau),
\end{equation}
the function $ \Phi_j(t, \tau) $ is defined by~\eqref{function_Phi};
$ \eta_j \in {\widetilde G}_1 $ and additionally
\[
\lim_{\tau\to -\infty} \eta_j(t, \tau) = 1 , \quad
\psi_j \in {\widetilde G_0} .
\]

To prove relation \eqref{vyglyad} we substitute it into differential equation \eqref{sing_part_j_1}, take into account operator equation \eqref{equation_L} and conclude that the function $\tau\mapsto \psi_j (t, \tau ) $ has to satisfy the inhomogeneous equation
\begin{equation} \label{ad_eq_2}
L \psi_j = \Phi_j - \nu_j L \eta_j,
\end{equation}
where the right-side function $ \Phi_j - \nu_j L \eta_j \in {\mathcal{S}({\mathbb R})} $. 
	
In virtue of Theorem~\ref{thm: th_1} equation \eqref{ad_eq_2} has a solution in the space $ {\widetilde G_0} $ if and only if the following orthogonality condition
\begin{equation} \label{orth_cond_1}
\int_{-\infty}^{+\infty} \left( \Phi_j - \nu_j L \eta_j \right) v_{0\tau} d \tau = 0 
\end{equation}
holds.

Finally, upon integrating \eqref{orth_cond_1} and invoking \eqref{function_Phi}, \eqref{sing_part_j2_soliton_sol_1}, and \eqref{equation_L}, the statement of Lemma~\ref{lem: lemma_2} follows.
\end{proof}

\begin{remark}\label{rem: remark_2}
The representation of the function $ v_j(t, \tau) $ in form \eqref{vyglyad} is used to extend it from the curve $ \Gamma $ to the desired domain.
\end{remark}

Of particular interest is the case when equation \eqref{sing_part_j2_soliton_sol_1} admits solutions whose properties are entirely analogous to those of the main singular term, that is, solutions belonging to the space $ {\widetilde G}_0 $. The conditions ensuring the existence of such solutions are provided by the following lemma.

\begin{lemma} \label{lem: lemma_3}
 Let the right-side function $ {\cal F}_j $ in \eqref{sing_part_j2_soliton_sol_1} belong to the space $ {\widetilde G_{ 0}} $, inequalities \eqref{cond_w} and the orthogonality condition \eqref{ort_cond} be satisfied.
Then the solution $ v_j $ of equation \eqref{sing_part_j2_soliton_sol_1} is an element of the space $ {\widetilde G_{ 0}} $ if and only if the condition
\begin{equation}\label{ort_cond_2_0}
\lim_{\tau \to -\infty} \Phi_j(t, \tau) = 0
\end{equation}
holds.
\end{lemma}

\begin{proof}
The lemma follows directly from representation~\eqref{vyglyad}. Indeed, the relation \eqref{ort_cond_2_0} means that
$ E_j(t) = 0 $ in \eqref{function_Phi}. This equality with formula \eqref{vyglyad}--\eqref{ad_eq_2} yields the inclusion $ v_j \in {\widetilde G_0}$.
\end{proof}

\subsubsection{Analytical representation of higher terms}

The function $ v_j (t, \tau) $, $ j = 0, \dots, N $, gives the value of the function $ V_j (x, t, \tau) $ on the discontinuity curve $ \Gamma $.
We now proceed to extend $ V_j(x, t, \tau) $ from the curve $ \Gamma $ into its neighborhood. Since the function $ v_0(t, \tau) \in {\widetilde G}_0 $, we can do it as follows
\begin{equation} \label{prolong_V_0}
V_0(x, t, \tau) = v_0(t, \tau).
\end{equation}

For $ j = 1, \dots, N $, the extension of the function $ v_j(t, \tau) $ depends on its properties as $ \tau \to - \infty $. If $ v_j(t, \tau)\in {\widetilde G_0} $, then its continuation can be written analogously to \eqref{prolong_V_0} as
\begin{equation} \label{prolong_V_j_0}
	V_j(x, t, \tau) = v_j(t, \tau).
\end{equation}

In the opposite case, i.e., when condition \eqref{ort_cond_2_0} is not true, we make use of representation \eqref{vyglyad}, and the prolongation is constructed by
\begin{equation} \label{prolong_V_j}
	V_j(x, t, \tau) = u_j^- (x, t) \eta_j(t, \tau) + \psi_j(t, \tau),
\end{equation}
where the functions $ \eta_j(t, \tau) $, $ \psi_j(t, \tau) $ are defined via formulae \eqref{vyglyad} and \eqref{nu}, while the function $ u_j^- (x, t) $ is a solution of the auxiliary Cauchy problem of the form
\begin{align} \label{prolong}
	\Lambda u_j^- (x, t) & = f_j^- (x, t),
\\
\label{prolong_0}
	u_j^- (x, t)\bigr|_{\Gamma} & = \nu_j(t)
\end{align}
with linear differential operator
\begin{equation}\label{operator_Lambda}
	\Lambda = a_0(x, t) \frac{\partial}{\partial t} + b_0(x, t) u_0(x, t) \frac{\partial}{\partial x} + b_0(x, t) u_{0x}(x, t).
\end{equation}

The right-side function $ f_j^-(x, t) $ is defined recursively after the functions $ u_1^-, \dots , u_{j-1}^- $ have been constructed, while the initial data in \eqref{prolong_0} is determined by \eqref{nu}. In particular,
\begin{align*}
f_1^-(x, t) & = 0, \\
f_2^-(x, t) & = - a_1(x, t) \frac{\partial u_1^-}{\partial t} - b_0(x, t) [ u_1 + u_1^- ] \frac{\partial u_1^-}{\partial x} \\
 & - b_1(x, t) \frac{\partial}{\partial x} ( u_1^- u_0 ).
\end{align*}

The differential equation \eqref{prolong} is deduced after substituting the representation \eqref{prolong_V_j} into equation \eqref{CHolm_vc} and limiting as variable $ \tau $ tends to $ - \infty $. The initial condition \eqref{prolong_0} follows from the representation \eqref{vyglyad}.

The Cauchy problem \eqref{prolong}, \eqref{prolong_0} has a solution at least in a neighborhood of the curve $ \Gamma $, since the curve $ \Gamma = \{(x, t) \mid x = \varphi(t),\, t \in [0; T] \} $ is transversal to the characteristics of the operators $ \Lambda $ in virtue of condition \eqref{cond_w}.

Depending on whether condition \eqref{ort_cond_2_0} is satisfied, formulas \eqref{prolong_V_j_0} and \eqref{prolong_V_j} provide the analytical representation of the singular terms $ V_j(x,t,\tau) $, $ j = 1, 2, \ldots $, which, together with $ V_0(x, t, \tau) $ defined by \eqref{prolong_V_0} and the regular terms $ u_j(x, t) $, $ j = 0, 1, \ldots $, form the asymptotic soliton-like solution to equation \eqref{CHolm_vc}.
In both cases, the asymptotic solution satisfies the equation with the same precision. This is confirmed by the following statements, formulated on the basis of the procedure for constructing regular and singular terms.

\begin{theorem}\label{thm: th_2}
Assume the following conditions:
\begin{enumerate}
\item[1.] The functions $ a_j(x, t) $, $ b_j(x, t) \in C^{\infty} ({\mathbb R}\times [0;T])$, $ j = 0, 1, \dots, N $, and
$$
a_0(x, t) b_0(x, t) \not= 0 \quad \mbox{for all} \quad (x, t) \in {\mathbb R} \times [0;T];
$$
\item[2.] The inequality \eqref{cond_w} is fulfilled and $ \varphi'(t) \not= 0 $ for all $ t \in [0;T]$;
\item[3.] The orthogonality conditions \eqref{ort_cond_20} hold for all $ j=1,2, \ldots, N $;
\item[4.] Conditions \eqref{ort_cond_2_0} are true for all $ j=1,2, \ldots, N $.
\end{enumerate}
Then the function \eqref{soliton_one-phase} represents the $N$-th approximation of the asymptotic soliton-like solution to equation \eqref{CHolm_vc} and satisfies this equation on the set $ {\mathbb R} \times [0;T] $ with an asymptotic accuracy $ O\bigl(\varepsilon^N\bigr) $.
\end{theorem}

The form of the analytic extension essentially depends on the validity of the condition~\eqref{ort_cond_2_0}. If this condition is not violated -- equivalently if the singular term in the asymptotic expansion do not belong to the space $ G_0 $ -- the following theorem can be established.

\begin{theorem}\label{thm: th_3}
Let the following assumptions be supposed:
\begin{enumerate}
\item[1.]
The conditions $1$ -- $3$ of Theorem~\ref{thm: th_2} are true;
\item[2.] The Cauchy problem \eqref{prolong}, \eqref{prolong_0} has a solution on the set
$$
\{(x, t) \in {\mathbb R} \times[0;T]\mid x - \varphi(t) \le 0 \} .
$$
\end{enumerate}
Then the function \eqref{soliton_one-phase} represents the $N$-th approximation of the asymptotic soliton-like solution to equation \eqref{CHolm_vc} and satisfies this equation on the set $ {\mathbb R} \times [0;T] $ with an asymptotic accuracy $ O\bigl(\varepsilon^N\bigr) $.
\end{theorem}

\begin{proof}

The ideas behind the proofs of Theorems~\ref{thm: th_2} and \ref{thm: th_3} are similar to those used in the proof of Theorem~1 in \cite{Sam_MMC_2021_2}. Therefore, we do not repeat them here in details and instead provide a short outline of the main arguments.

To prove Theorem~\ref{thm: th_2}, we analyze the residual function and obtain appropriate estimates for it. In doing so, we use the fact that the singular terms satisfy equations~\eqref{sing_part_02_soliton_sol_1} and \eqref{sing_part_j2_soliton_sol_1}, and that they decay rapidly with respect to the variable
$ \tau $. We also apply the local representations (Taylor expansions) of the coefficients of the functions $ a(x, t, \varepsilon) $, $ b(x, t, \varepsilon) $ in a neighborhood of the curve~$ \Gamma $.

In the proof of Theorem~\ref{thm: th_3}, we additionally use the properties of the functions belonging to the spaces $ {\widetilde G}_1 $ and $ \widetilde G_0 $.
\end{proof}

It follows from the procedure of constructing the terms of the asymptotic expansion \eqref{soliton_one-phase} that the main terms in the asymptotic series of the coefficients $ a(x,t, \varepsilon) $, $ b(x, t, \varepsilon) $ play a crucial role in determining the asymptotic solution. In particular, $ a_0(x, t) $, $ b_0(x, t) $ are incorporated into main regular and singular terms capturing soliton features of the searched solutions.

Another key aspect of the procedure concerns the phase function $ \varphi(t) $, whose determination is a nontrivial problem and requires a specialized approach for each wave model. For instance, in the cases of the vcKdV \cite{Sam_2005} and vcBBM \cite{Sam_JMP} equations, the phase function is obtained from an ordinary differential equation derived from the corresponding orthogonality condition, whereas for the vcmCH equation \cite{SamBrandSam}, it is determined during the construction of the main singular term using the $ G'/G $ method. For the vcCH equation, the phase function have to satisfy condition \eqref{cond_w} and the differential equation deduced from the orthogonality condition \eqref{ort_cond}, ensures the existence of the first singular term with the required properties.

\subsubsection{Example 1 }

To demonstrate the application of the approach described above, we consider the construction of the first two terms of the asymptotic one-phase soliton-like solution \eqref{soliton_one-phase} for particular case of equation \eqref{CHolm_vc} with coefficients specified as
\begin{align}\label{coeff_ex_1}
\begin{aligned}
& a_0(x,t) = b_0(x, t) = 1, \quad a_1(x,t) = x^2 + 1, \\
& b_1(x,t) = a_2(x, t) = b_2(x,t) = \dots = 0.
\end{aligned}
\end{align}

Thus, we deal with vcCH equation of the form
\begin{equation}\label{example_3_s}
	\bigl[1 + \varepsilon\bigl(x^2 + 1\bigr)\bigr] u_t - \varepsilon^2 u_{xxt} + u u_x - 2 \varepsilon^2 u_x u_{xx} - \varepsilon^2 uu_{xxx} = 0.
\end{equation}

The required regular terms can be obtained directly from \eqref{reg_part_0} and \eqref{reg_part_j} using the coefficients specified in \eqref{coeff_ex_1}. The regular terms serve as the background component for the soliton-like solution, so to simplify the calculations we choose $ u_0(x, t) = 1 $ and $ u_1(x, t) = 1 $, which obviously satisfy the system mentioned above.

Since key point of constructing the asymptotic soliton-like solution is finding values of the singular terms on the discontinuity curve $ \Gamma $, we consider the equations for $ v_0(t, \tau) $ and $ v_1(t, \tau) $. From formula \eqref{main_term}, we obtain the main singular term
\begin{equation}\label{main_term-ex}
v_0(t, \theta) = \frac{15}{4} \bigl(5 \cosh^2 \theta - 3 \sinh^2 \theta\bigr)^{-1},	
\end{equation}
and relation between $ \tau $ and $ \theta $ in the form
\begin{equation}\label{theta-ex}
 \tau = 2 \sqrt{5} \theta + \sqrt{3} \ln{\frac{\cosh(\theta - \theta_0)}{\cosh(\theta + \theta_0)}},
\end{equation}
where
$$
\theta_0 = \tanh^{-1} \sqrt{\frac{3}{5}}.
$$

According Lemma~\ref{lem: lemma_1}, representation \eqref{theta-ex} defines a bijective mapping $ \theta \mapsto \tau $ on the set $ {\mathbb R} $, which is illustrated in Figure~\ref{fig: 1}.

\begin{figure}[h]
\centering
\includegraphics[scale=0.7]{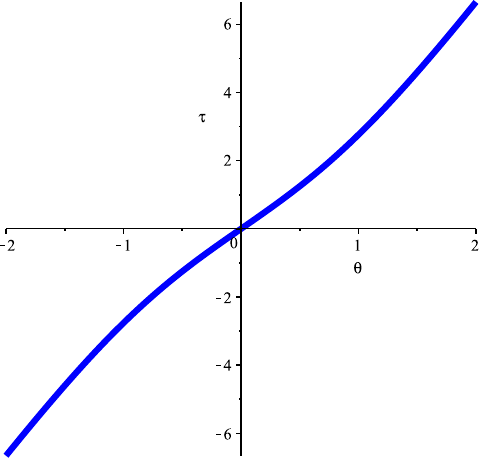}
\caption{The graph of function \eqref{theta-ex}. \label{fig: 1}}
\end{figure}

Relation \eqref{sing_part_j_1} yields the differential equation satisfied by the function $ v_1(t, \tau) $:
\begin{align*}
 ( \varphi' - v_0 ) \frac{\partial^3 v_1} {\partial\tau^3} + ( 1 - \varphi' ) \frac{\partial v_1}{\partial \tau} +
\frac{\partial }{\partial \tau } ( v_0 v_1 ) \\
- 2 \frac{\partial }{\partial \tau } \left( \frac{\partial v_0}{\partial \tau} \frac{\partial v_1}{\partial \tau} \right) - \frac{\partial^3 v_0}{\partial \tau^3} v_1 = {\cal F}_1
\end{align*}
with
$$
{\cal F}_1 = (1 + \varphi^2) \varphi' \frac{\partial v_0}{\partial \tau}.
$$

By virtue of condition \eqref{cond_w}, the phase function $ \varphi = \varphi(t) $ has to satisfy an inequality
$$
1 < \varphi'(t) < \frac{3}{2}.
$$

In this case, when constructing the first singular terms, there are no additional constraints on the phase function. Therefore, we can take $ \varphi(t) = {5 t}/{4} $, $ t\in {\mathbb R} $.

It is easy to verify that the orthogonality condition \eqref{ort_cond} for $ j=1 $ holds and the first singular term $ v_1(t, \theta) $ is written as
\begin{align}
v_1(t, \theta) ={}& \frac{75}{64} \left( 16 + 25 t^2\right) \left(3 + 2 \cosh^2 \theta \right)^{-3} \nonumber\\
& \times \left[ 25 - 60 \theta \tanh{\theta} - \left( 35 + 6 \cosh^2 \theta \right)\sinh^2 \theta \right.\nonumber\\
& \left. +2 \theta \left( 3 + \cosh^2 \theta \right) \sinh 2 \theta \right].\label{first_term-ex}	
\end{align}

We ensure that the conditions of Lemma~\ref{lem: lemma_3} are satisfied because the function
$$
\Phi_1(t, \tau) = \int_{ -\infty}^\tau {\cal F}_1 (t, \xi) d \xi + E_1(t) = (1 + \varphi^2) \varphi' v_0
$$
meets condition \eqref{orth_cond_1}, which guarantees that $ v_1(t, \tau) \in {\widetilde G}_0 $. Taking into account inclusion $ v_0(t, \tau) \in {\widetilde G}_0 $, this allows us to consider the extensions of $ v_0(t, \tau) $ and $ v_1(t, \tau) $ as
$$
V_0(x, t, \tau) = v_0(t, \tau), \quad V_1(x, t, \tau) = v_1(t, \tau).
$$

Thus, the first asymptotic approximation is given as
$$
Y_1(x, t, \tau) = u_0(x, t) + v_0(t, \tau) + \varepsilon [u_1(x, t) + v_1(t, \tau) ],
$$
the graphs of which are presented in Figures~\ref{fig: 2}--\ref{fig: 4}.

\begin{figure}[ht]
\centering
\includegraphics[scale=0.51]{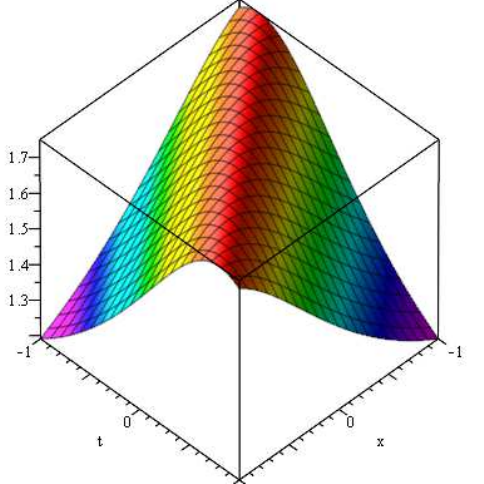} \quad
\includegraphics[scale=0.51]{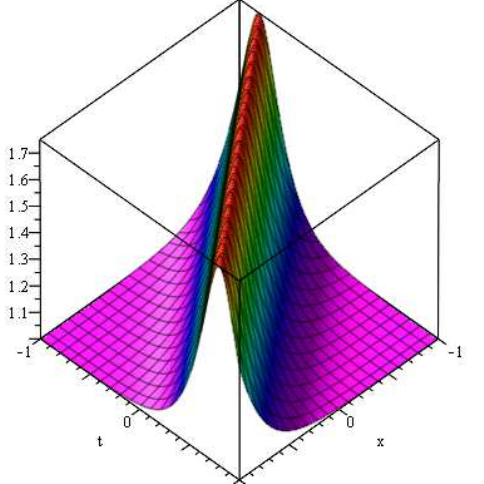}
\caption{The main term 
$ u_0(x, t) + v_0 (t, \tau)$ as $\varepsilon=0.5$ (at the left) and $\varepsilon=0.1$ (at the right).} \label{fig: 2}
\end{figure}

\begin{figure}[ht]
\centering
\includegraphics[scale=0.51]{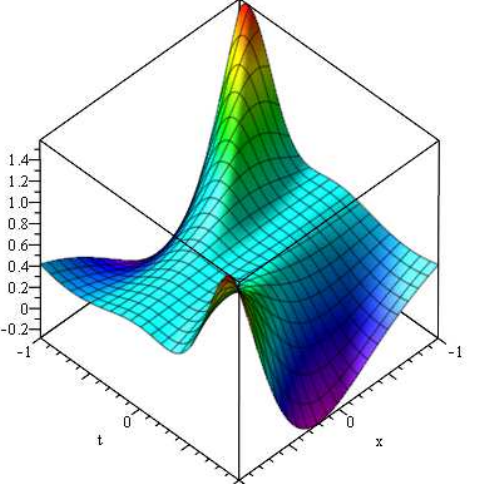} \quad
\includegraphics[scale=0.51]{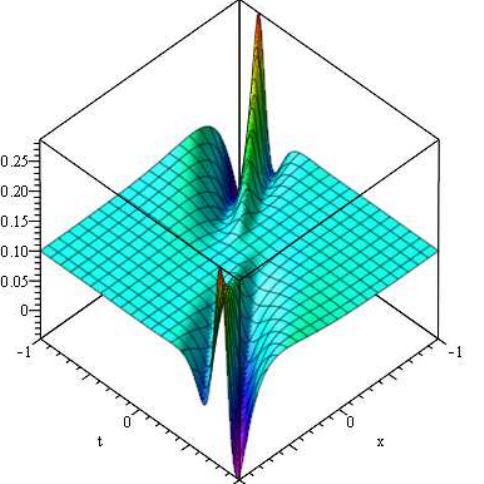}
\caption{The first term 
$ \varepsilon (u_1(x, t) + v_1 (t, \tau) )$ as $\varepsilon=0.5$ (at the left) and $\varepsilon=0.1$ (at the right).} \label{fig: 3}
\end{figure}

\begin{figure}[ht]
\centering
\includegraphics[scale=0.51]{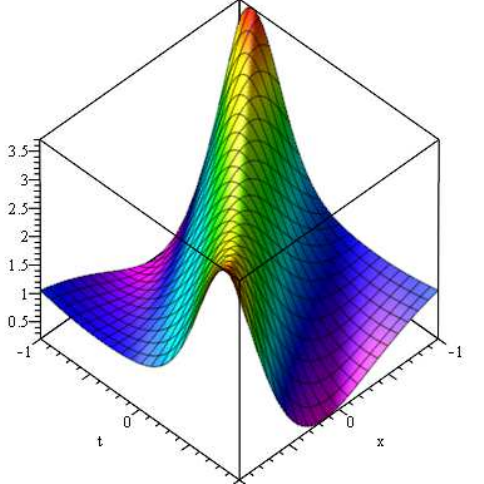} \quad
\includegraphics[scale=0.51]{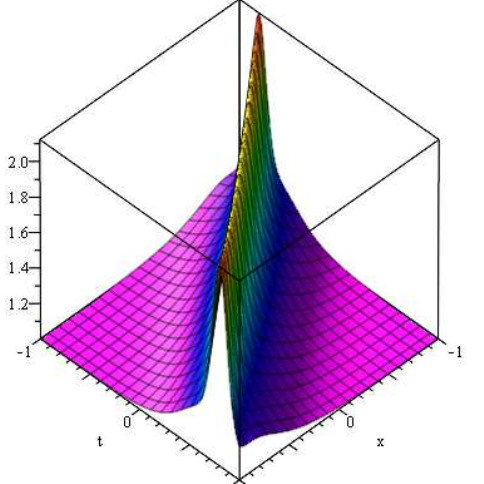}
\caption{The first asymptotic approximation
$ Y_1 (x, t, \tau)$ as $\varepsilon=1$ (at the left) and $\varepsilon=0.5$ (at the right).} \label{fig: 4}
\end{figure}

\section{Two-phase soliton-like solutions}\label{section3}

The CH equation is known to possess multi-phase soliton solutions that asymptotically split into one-phase soliton components at infinity \cite{Parker_2005}. It is natural to consider the problem of the asymptotic soliton-like solutions for the vcCH equation that exhibit properties analogous to those of the multi-phase soliton solutions of the original CH equation.

In this paper, we focus on the construction of asymptotic two-phase soliton-like solutions. As in the one-phase case, the asymptotic two-phase soliton-like so\-lu\-tion is \linebreak sought as a sum of the regular and singular parts of the asymp\-totic expansion.

While the regular part of the asymptotics of the desired solution can be determined with arbitrary accuracy, constructing the singular part faces technical difficulties that, so far, remain insurmountable. These difficulties stem from the need to determine higher-order singular terms as solutions of third-order partial differential equations with variable coefficients within a specialized functional space.

Consequently, it is currently possible to determine on\-ly the main singular term of the asymptotics, which itself satisfies a third-order nonlinear partial differential equation. Notably, this result -- the construction of the main singular term -- represents a significant achievement, since it is precisely this term that plays a decisive role in capturing the soliton characteristics of the solution.

To date, two-phase and multi-phase solutions for in\-teg\-rable-type equations with variable coefficients and singular perturbations have been constructed only in the case of the KdV equation \cite{Sam_2008, Sam_2012_1, Sam_2012_2} and KdV-type equations \cite{Maslov_book}. The success in constructing $m$-soliton solutions ($m\ge 2 $) for the vcKdV equation was made possible by the presence of only a single nonlinear term and its flexibility with respect to the coefficients, which is reflected in its calibration equivalence in the case of constant coefficients. In contrast, the corresponding problem for the vcCH equation lacks such properties and, consequently, requires additional constraints.

We introduce suitable functional spaces whose elements reflect the characteristic behavior of two-phase soliton solutions.
By $ G_2 $ we note the space of the functions $ f = f(x, t, \tau_1, \tau_2) \in C^\infty
({\mathbb{R}} \times [0; T] \times {\mathbb R} \times {\mathbb R}) $, for which there exist functions $ f_1^{\pm} = f_1^{\pm} (x,t, \tau_2) $, $ f_2^{\pm} = f_2^{\pm} (x,t, \tau_1) \in G_0 $, such that for all non-negative integers $ n $, $ p $, $ q $, $ \alpha $, $\beta $ uniformly with respect to $ (x,t) \in K $ on any compact set $ K \subset {\mathbb{R}} \times [0; T] $ we have the relations
$$
\lim_{\tau_k \to \pm \infty} \tau_k^{ n} \frac {\partial^{ p}} {\partial x^p}\frac{\partial^{ q}}{\partial t^q}
\frac{\partial^{ \alpha}} {\partial \tau_1^{\alpha}} \frac{\partial^{ \beta}} {\partial \tau_2^{\beta}} \left(
f - f_k^{\pm} \right) = 0, \quad k =1, 2.
$$

Analogously to the space $ {\widetilde G}_1 $, we define $ {\widetilde G}_2  \subset G_2 $ as a space of functions $ f = f(t, \tau_1, \tau_2) \in C^\infty ([0; T] \times {\mathbb R} \times {\mathbb R}) $, for which there exist functions $ f_1^{\pm} = f_1^{\pm} (t, \tau_2) $,
$ f_2^{\pm} = f_2^{\pm} (t, \tau_1) \in {\widetilde G}_0 $ such that, for all non-negative integers $ n $, $ q $, $ \alpha $, $\beta $ uniformly in $ t \in [0; T] $, the following relations hold:
$$
\lim_{\tau_k \to \pm \infty} \tau_k^{ n} \frac{\partial^{q}}{\partial t^q}
\frac{\partial^{ \alpha}} {\partial \tau_1^{\alpha}} \frac{\partial^{ \beta}} {\partial \tau_2^{\beta}} \left(
f - f_k^{\pm} \right) = 0, \quad k =1, 2.
$$

\begin{definition}[\cite{Maslov_book, Sam_2008}] \label{def_2} The function $ u(x, t, \varepsilon) $, where $ \varepsilon $ is a small parameter,
is called two-phase asymptotic so\-li\-ton-like function, if for any integer $ N \ge 0 $ it can be written as
\begin{equation}
\label{two-phase-soliton-like} u(x, t, \varepsilon) = \sum\limits_{j=0}^N \varepsilon^j [u_j(x, t) + V_j(x, t, \tau_1, \tau_2 ) ]
+ O\bigl(\varepsilon^{N+1}\bigr),
\end{equation}
where
$$
\tau_1 = \frac{x-\varphi_1(t)}{\varepsilon}, \quad \tau_2 = \frac{x-\varphi_2(t)}{\varepsilon},
$$
the functions $ V_j \in G_2 $, $ j=0,\dots,N $, and the phase functions $ \varphi_k = \varphi_k(t) $, $ t \in [0; T] $, $ k = 1, 2 $, satisfy the initial condition $ \varphi_1(0) = \varphi_2(0) $.
\end{definition}

The two soliton-like waves described by the two-phase soliton-like solution exhibit a nonlinear interaction, after which they separate and resume independent propagation. For convenience, we suppose that the point $ t = 0 $ corresponds to this moment.

The functions $ u_j(x, t) $, $ j = 0, 1, \ldots $, form the regular part of asymptotics \eqref{two-phase-soliton-like} and they represent the background components of the wave propagation. As before, these functions satisfy the system \eqref{reg_part_0}, \eqref{reg_part_j} solvability of which are discussed in Section~\ref{section2.2}.

In contrast, the functions $ V_j $, $ j = 0, 1, \ldots $, form the singular part of the asymptotics and are introduced to capture the soliton characteristics of the two-phase soliton-like solution. They are defined as solutions of certain partial differential equations and are written through the phase functions, $ \varphi_k $, $ k =1, 2 $, each of which determines a phase curve
$$
\Gamma_k = \{(x, t) \in {\mathbb R}\times [0;T]\mid x - \varphi_k (t) = 0 \}, \quad k =1,2.
$$

\subsection{The main singular term}

Using standard techniques, we obtain the following differential equation for the main singular term:
\begin{align}
&  - a_0(x, t) \left[\varphi_1'(t) \frac{\partial V_0}{\partial \tau_1} + \varphi_2'(t) \frac{\partial V_0}{\partial \tau_2} \right] \nonumber + b_0 (x, t) u_0(x, t) \left[\frac{\partial V_0}{\partial \tau_1}+ \frac{\partial V_0}{\partial \tau_2} \right]
\nonumber \\
& + \left[ \varphi_1'(t) \frac{\partial^3 V_0}{\partial \tau_1^3} + \left(2 \varphi_1'(t) + \varphi'_2(t) \right) \frac{\partial^2 V_0}{\partial \tau_1^2\partial \tau_2} \right.
\nonumber \\
& \qquad \left. + \left(\varphi_1'(t) + 2 \varphi'_2(t) \right) \frac{\partial^3 V_0}{\partial \tau_1 \partial \tau_2^2} + \varphi'_2(t) \frac{\partial^3 V_0}{\partial \tau_2^3} \right]
\nonumber \\
& + b_0 (x, t) V_0 \left[ \frac{\partial V_0}{\partial \tau_1} + \frac{\partial V_0}{\partial \tau_2} \right]
\nonumber \\
& \qquad - 2 \left[ \frac{\partial V_0}{\partial \tau_1} + \frac{\partial V_0}{\partial \tau_2} \right]
 \left[\frac{\partial^2 V_0}{\partial \tau_1^2} + 2 \frac{\partial^2 V_0}{\partial \tau_1\partial \tau_2} + \frac{\partial^2 V_0}{\partial \tau_2^2} \right]
\nonumber \\
& - V_0 \left[\frac{\partial^3 V_0}{\partial \tau_1^3} + 3 \frac{\partial^2 V_0}{\partial \tau_1^2\partial \tau_2} + 3 \frac{\partial^3 V_0}{\partial \tau_1 \partial \tau_2^2}+ \frac{\partial^3 V_0}{\partial \tau_2^3} \right] = 0.  \label{sing_term_0_2phase}
\end{align}
We are interesting in a particular solution which due to $ V_0 \in G_2 $ rapidly vanishes outside a certain neighborhood of the phase curves. As a consequence we may study the equation along the curves $ \Gamma_1 $, $ \Gamma_2 $. Noting also that equation~\eqref{sing_term_0_2phase} is symmetric with respect to the variables $ \tau_1 $, $ \tau_2 $ that allows to seek its particular solutions in the form of functions that possess symmetry with respect to $ \tau_1 $, $ \tau_2 $. Complexity of the solving equation forces us to restrict ourselves to the assumption that its coefficients $a_0(x, t)$, $b_0(x, t)$ and $ u_0(x, t) $ are constant along the phase curves, i.e.:
\begin{align}\label{cond_conj}
\begin{aligned}
& a_0(\varphi_1, t) = a_0(\varphi_2, t) = 1, \\ & b_0(\varphi_1, t) = b_0(\varphi_2, t) = 3, \\
& u_0(\varphi_1, t) = u_0(\varphi_2, t) = u_0 (t),
\end{aligned}
\end{align}
where $ \varphi_1 = \varphi_1 (t) $, $ \varphi_2 = \varphi_2 (t) $, $ t \in [0; T] $.

Then, relation \eqref{sing_term_0_2phase} yields an equation
\begin{align}
& \left[- \varphi_1' \frac{\partial V_0}{\partial \tau_1} - \varphi_2' \frac{\partial V_0}{\partial \tau_2} \right]
+ 3 u_0 \left[\frac{\partial V_0} {\partial \tau_1} + \frac{\partial V_0}{\partial \tau_2} \right]  \nonumber\\
&\qquad
+ \left[ \varphi_1' \frac{\partial^3 V_0}{\partial \tau_1^3} + \left(2 \varphi_1' + \varphi'_2 \right) \frac{\partial^2 V_0}{\partial \tau_1^2\partial \tau_2} \right. + \left. \left(\varphi_1' + 2 \varphi'_2 \right) \frac{\partial^3 V_0}{\partial \tau_1 \partial \tau_2^2} + \varphi'_2 \frac{\partial^3 V_0}{\partial \tau_2^3} \right] \nonumber\\
& \qquad + 3 V_0 \left[ \frac{\partial V_0}{\partial \tau_1} + \frac{\partial V_0}{\partial \tau_2} \right] - 2 \left[ \frac{\partial V_0}{\partial \tau_1} + \frac{\partial V_0}{\partial \tau_2} \right] \left[\frac{\partial^2 V_0}{\partial \tau_1^2} + 2 \frac{\partial^2 V_0}{\partial \tau_1\partial \tau_2} + \frac{\partial^2 V_0}{\partial \tau_2^2} \right]\nonumber \\
& \qquad - V_0 \left[\frac{\partial^3 V_0}{\partial \tau_1^3} + 3 \frac{\partial^2 V_0}{\partial \tau_1^2\partial \tau_2} + 3 \frac{\partial^3 V_0}{\partial \tau_1 \partial \tau_2^2}+ \frac{\partial^3 V_0}{\partial \tau_2^3} \right] = 0\label{main_term2phase_sol}
\end{align}
with $ u_0 = u_0(t) $.

By introducing new variables $\xi $, $ \eta $ via formulas
$$
\tau_k = \xi - \varphi_k' \eta, \quad k = 1, 2,
$$
\eqref{main_term2phase_sol} is reduced to the following equation
\begin{equation}\label{main_term_two_phase_n-var}
\frac{\partial V_0}{\partial \eta} + 3 u_0 \frac{\partial V_0}{\partial \xi} + \frac{\partial^3 V_0}{\partial \xi^2 \partial \eta} = - 3 V_0 \frac{\partial V_0}{\partial \xi} + 2 \frac{\partial V_0}{\partial \xi} \frac{\partial^2 V_0}{\partial \xi^2} + V_0 \frac{\partial^3 V_0}{\partial \xi^3},
\end{equation}
that is in fact the Camassa--Holm equation.
Its two-soliton solution is written in the implicit form as \cite{Johnson}
\begin{align}
& V_0 = 12 u_0 \left(\frac{ \mu_1^2}{(1-\mu_1^2)^2} E_1 + \frac{ \mu_2^2}{(1-\mu_2^2)^2} E_2 \right. + \frac{2 (\mu_1 - \mu_2)^2 (1-\mu_1^2 \mu_2^2)}{(1-\mu_1^2)^2 (1-\mu_2^2)^2} E_1 E_2 \nonumber \\
& \left. + \frac{(\mu_1 - \mu_2)^2}{(\mu_1 + \mu_2)^2} \left[\frac{\mu_2^2}{(1-\mu_2^2)^2} E_1 + \frac{\mu_1^2}{(1-\mu_1^2)^2} E_2 \right] E_1 E_2 \right) \nonumber \\
& \times \left(1 + \frac{2(1+\mu_1^2)}{(1-\mu_1^2)} E_1 + \frac{2(1+\mu_2^2)}{(1-\mu_2^2)} E_2 + E_1^2 + E_2^2 \right.\nonumber \\
& + \left. 4 \frac{(\mu_1^2 + \mu_2^2)(1 + \mu_1^2 \mu_2^2) + (\mu_1^2 - \mu_2^2)^2 - 4 \mu_1^2 \mu_2^2}{(1 - \mu_1^2)(1 - \mu_2^2)(\mu_1 + \mu_2)^2} E_1 E_2 \right.\nonumber \\
& + \frac{2(\mu_1 - \mu_2)^2}{(\mu_1 + \mu_2)^2} \left[ \frac{(1 + \mu_2^2)}{(1 - \mu_2^2)} E_1 + \frac{(1 + \mu_1^2)}{(1 - \mu_1^2)} E_2 \right] \left. E_1 E_2 \left(\frac{\mu_1 - \mu_2}{\mu_1 + \mu_2} \right)^4 E_1^2 E_2^2\right)^{-1}, \label{sol_sol_CH_equation}
\end{align}
where
\begin{align}
\mu_k^2 ={}& 1 - 3 \sqrt{6} \frac{ u_0 \sqrt{u_0}}{\varphi'_k}, \quad \mu_k > 0, \quad E_k(t, \tau_1, \tau_2) = \exp{(\delta_k)},
\nonumber\\
\label{delta_k}
\delta_k ={}& \frac{\sqrt{6}}{3 \sqrt{u_0}} \mu_k \left(-\tau_k + \ln \frac{\Delta_1}{\Delta_2}\right),
\\
{\Delta}_1 ={}& ( 1 - \mu_1) ( 1 - \mu_2) + ( 1 + \mu_1) ( 1 + \mu_2) E_1 \nonumber\\
& + ( 1 - \mu_1) ( 1 - \mu_2) E_2
\nonumber\\
&
+ ( 1 + \mu_1) ( 1 + \mu_2) \frac{(\mu_1 - \mu_2)^2}{(\mu_1 + \mu_2)^2} E_1 E_2,
\\
{\Delta}_2 ={}& ( 1 + \mu_1) ( 1 + \mu_2) + ( 1 - \mu_1) ( 1 + \mu_2) E_1
\nonumber\\
&
+ ( 1 + \mu_1) ( 1 - \mu_2) E_2 \nonumber\\
& + ( 1 - \mu_1) ( 1 - \mu_2) \frac{(\mu_1 - \mu_2)^2}{(\mu_1 + \mu_2)^2} E_1 E_2,
\end{align}
and $ u_0 (t) $, $ \varphi_1 (t) $, $ \varphi_2 (t) $ for all $ t \in [0; T] $ satisfy the following inequalities:
\begin{equation} \label{cond_2sol_phase}
u_0 > 0, \quad 0 < \frac{u_0 \sqrt{u_0}}{\varphi'_k} < \frac{1}{3 \sqrt{6}}, \quad k = 1, 2.
\end{equation}

The function $ V_0 $ in \eqref{sol_sol_CH_equation} depends on the variables $t$, $\tau_1$ and $\tau_2$.
Unlike the previously considered the vcKdV model \cite{Sam_2008}, formula \eqref{sol_sol_CH_equation} for the main singular term in the case of the vcCH equation \eqref{CHolm_vc} can hardly be derived from relation \eqref{sol_sol_CH_equation} in an explicit form as a function of the phase variables $ \tau_1 $ and $ \tau_2 $, even if one attempts to choose suitable values of the parameters $ \mu_1 $ and $ \mu_2 $. This indicates a significant complexity in constructing its asymptotic two-phase soliton-like solutions. Thus we need to ensure that the function $ V_0 $ belongs to the space $G_2$. Since the expression for $V_0$ involves implicitly $\delta_1$ and $\delta_2$, the following lemma will be needed for this purpose.

\begin{lemma}\label{lem: lemma_4} Let inequalities \eqref{cond_2sol_phase} be fulfilled. Then $ \tau_k \to + \infty (-\infty) $ if and only if
$ \delta_k \to - \infty (+\infty) $ for $ k= 1, 2 $.
\end{lemma}

\begin{proof} The statement of Lemma~\ref{lem: lemma_4} directly follows from relation \eqref{delta_k} and from the asymptotic analysis of the expressions for $\Delta_1$ and $\Delta_2$ as $\delta_1 \to \infty $, $\delta_2 \to \infty$, either independently or simultaneously. In doing so, it is shown that the quantity $ \ln (\Delta_1/ \Delta_2) $ remains bounded.
\end{proof}

The inclusion $ V_0 \in {\widetilde G}_2 $ is ensured by the following statement.

\begin{lemma}\label{lem: lemma_5}Let inequalities \eqref{cond_2sol_phase} be satisfied. Then for any non-negative integer $ n $ the following relations be hold:
\begin{align*}
& \lim_{\tau_1 \to \pm \infty} \tau_1^n \left( V_0 (t, \tau_1, \tau_2) - f_1^{\pm}(t, \tau_2)\right) = 0,
\\
& \lim_{\tau_2 \to \pm \infty} \tau_2^n \left( V_0 (t, \tau_1, \tau_2) - f_2^{\pm}(t, \tau_1)\right) = 0,
\end{align*}
where
\begin{multline*}
f_k^{+}(t, \tau_s) =\frac{12 \mu_s^2 E_s u_0}{(1 - \mu_s^2)^2 (1 + E_s^2) + 2 (1 - \mu_s^4) E_s },
\\[5mm]
f_k^{-}(t, \tau_s) = \frac{12 (\mu_1^2 - \mu_2^2)^2 \mu_s^2 E_s u_0}{(1 - \mu_s^2)^2} \\
\times \frac{1}{(\mu_1 + \mu_2)^4 + 2 (\mu_1^2 - \mu_2^2)^2 (1 + \mu_s^2) E_s + (\mu_1 - \mu_2)^4 E_s^2},
\end{multline*}
$ k, s = 1, 2 $, and $ k + s = 3$.
\end{lemma}

\begin{proof}
Proving Lemma~\ref{lem: lemma_5} is done by means of the direct calculations with taking into account formula \eqref{sol_sol_CH_equation} for the function $ V_0 $.
\end{proof}

\subsection{Accuracy of the main approximation}

Recall that the function $ u_0(x, t) $, which is a solution of equation \eqref{reg_part_0}, represents the main regular term, while the main singular term of the asymptotic two-phase soliton-like solution $ V_0(t, \tau_1, \tau_2) $ is given implicitly by \eqref{sol_sol_CH_equation}, where
$$
\tau_k = - \frac{\sqrt{6 u_0}}{2 \mu_k} \delta_k + \ln \frac{\Delta_1}{\Delta_2}, \quad k=1,2.
$$

\begin{theorem}\label{thm: th_4}
Let be the following conditions take place:
\begin{enumerate}
\item[1.] The functions $ a_0 $, $ b_0 \in C^{(1)} ({\mathbb R}\times [0; T])$, $ j = 0, 1, \dots, N $, $ a_0(x, t) b_0(x, t) \not= 0 $ for all $ (x, t) \in {\mathbb R}\times [0; T] $;
\item[2.] The phase functions $ \varphi_1 $, $ \varphi_2 \in C^{(1)} ([0; T]) $ are such that initial conditions $ \varphi_1(0) = \varphi_2(0) = 0$ hold, and
\begin{align*}
a_0(\varphi_1(t), t) &= a_0(\varphi_2(t), t) = 1, \\
b_0(\varphi_1(t), t) &= b_0(\varphi_2(t), t) = 3, \\
u_0(\varphi_1(t), t) &= u_0(\varphi_2(t), t)
\end{align*}
for all $ t \in [0; T] $;
\item[3.] Inequalities \eqref{cond_2sol_phase} are true.
\end{enumerate}

Then the main approximation of the asymptotic two-phase soliton-like solution
$$
Y_0(x, t, \varepsilon) = u_0(x, t) + V_0(t, \tau_1, \tau_2),
$$
where the function $ V_0 (t, \tau_1, \tau_2) $ is written through \eqref{sol_sol_CH_equation}, satisfies the vcCH equation \eqref{CHolm_vc} with an accuracy $ O(1) $.
\end{theorem}

\begin{proof} By substituting $ Y_0(x, t, \varepsilon) $ into equation~\eqref{CHolm_vc} and taking into account~\eqref{main_term2phase_sol}, we reduce the proving Theorem~\ref{thm: th_4} to deriving an asymptotic estimate for the following residual function:
\begin{align*}
g_0(x, t, \varepsilon) = & - \frac{1}{\varepsilon} \left[ a_0(x, t) - a_0(t)\right] \left[ \varphi_1' \frac{\partial V_0}{\partial\tau_1} + \varphi_2' \frac{\partial
V_0}{\partial\tau_2} \right] \\
& + \frac{1}{\varepsilon} \left[b_0 (x, t) u_0(x, t) - b_0(t) u_0(t) \right] \left[ \frac{\partial V_0}{\partial \tau_1} + \frac{\partial V_0}{\partial \tau_2} \right] \\
& + \frac{1}{\varepsilon} \left[b_0 (x, t) - b_0(t) \right] V_0 \left[ \frac{\partial V_0}{\partial \tau_1} + \frac{\partial V_0}{\partial \tau_2} \right]
\end{align*}
as $ \varepsilon\to 0 $.

Taking into account the inequality \eqref{cond_conj} and the inclusion $ V_0 \in {\widetilde G}_2 $, we obtain
\begin{align*}
\left| a_0(x, t) - a_0(t) \right| \left|\frac{\partial V_0}{\partial\tau_1} \right|
& = \left| a_0(x, t) - a_0(\varphi_1(t),
t) \right| \left|\frac{\partial V_0}{\partial\tau_1} \right| \\
& \le \varepsilon C_0 \left|\frac{x - \varphi_1(t)}{\varepsilon}\right| \left|\frac{\partial
V_0}{\partial\tau_1} \right| \le \varepsilon C_1,
\end{align*}
where $ C_1 > 0 $ is a constant independent of $ \varepsilon$, whose value is determined only by the compact set $ K \subset {\mathbb R} \times [0; T] $, on which the inequality is evaluated.

In a similar way, we deduce
\begin{align*}
 & | a_0(x, t) - a_0(t) | \left|\frac{\partial V_0}{\partial\tau_2} \right| \\
 & \qquad = \left| a_0(x, t) - a_0(\varphi_2(t), t) \right| \left|\frac{\partial V_0}{\partial\tau_2} \right| \le \varepsilon C_2,
\\
& | b_0(x, t) - b_0(\varphi_1(t), t) | \left|\frac{\partial V_0}{\partial\tau_1} \right| \\
&
\qquad  + \left| b_0(x, t) - b_0(\varphi_2(t), t) \right| \left| \frac{\partial
V_0}{\partial\tau_2} \right| \le \varepsilon C_3 ,
\\
& | b_0(x, t)u_0(x, t) - b_0(\varphi_1(t), t)u_0(\varphi_1(t), t) | \left|\frac{\partial V_0}{\partial\tau_1} \right| \\
& \qquad + \left| b_0(x, t)u_0(x, t) - b_0(\varphi_2(t), t)u_0(\varphi_2(t), t) \right| \left| \frac{\partial V_0}{\partial\tau_2} \right| \le \varepsilon C_4,
\end{align*}
where $ C_2 , C_3 , C_4 > 0 $ are constants depending only on the compact set $ K $.

This completes the proof of the theorem.
\end{proof}

\subsection{Example 2} \label{subsection_ example 2}
To present the result of applying the algorithm within the framework of Theorem~\ref{thm: th_4}, we consider the vcCH equation \eqref{CHolm_vc} with coefficients given as \cite{Popovych_2012}
$$
a_0(x, t) = \exp{\bigl( x^2 - 36 xt + 288 t^2\bigr)} , \quad b_0(x, t) = 3 + \varepsilon t^2,
$$
which leads to the following equation:
\begin{align}
& \exp{\bigl(x^2 - 36 x t + 288 t^2\bigr)} u_t - \varepsilon^2 u_{t x x} \nonumber\\
& \qquad{} + \bigl(3 + \varepsilon t^2\bigr) u u_x - 2 \varepsilon^2 u_x u_{xx} - \varepsilon^2 u u_{xxx} = 0.\label{ex_2_sol}
\end{align}

We take $ u_0(x, t) = 1 $ as the main regular term. For the phase functions, we choose $ \varphi_1(t) = 12 t $ and $ \varphi_2(t) = 24 t $, which ensures that the conditions of Theorem~\ref{thm: th_4} and condition~\eqref{cond_2sol_phase} are satisfied for all $ t $.

The main singular term is given by formula \eqref{sol_sol_CH_equation} with the parameters $ \mu_1 $ and $ \mu_2 $ as
$$
\mu_1 = \frac{\sqrt{4 - \sqrt{6}}}{2}, \quad \mu_2 = \frac{\sqrt{8 - \sqrt{6}}}{2\sqrt{2}}.
$$

Additionally,
\begin{align*}
\delta_1 = \frac{\sqrt{4 - \sqrt{6}}}{\sqrt{6}} \left( - \tau_1 + \ln \frac{\Delta_1}{\Delta_2}\right), \\
\delta_2 = \frac{\sqrt{8 - \sqrt{6}}}{\sqrt{6}} \left( - \tau_2 + \ln \frac{\Delta_1}{\Delta_2}\right), \\
E_k (t, \tau_1, \tau_2) = \exp{(\delta_k)}, \quad k =1, 2.
\end{align*}

The resulting main approximation of the two-phase soliton-like solution is therefore global in~$ t $ and~$ x $~\cite{Johnson}.

\begin{figure}[ht]
\centering
\includegraphics[scale=0.51]{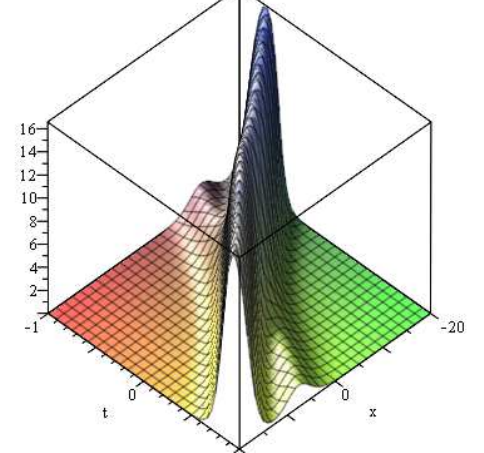} \quad
\includegraphics[scale=0.51]{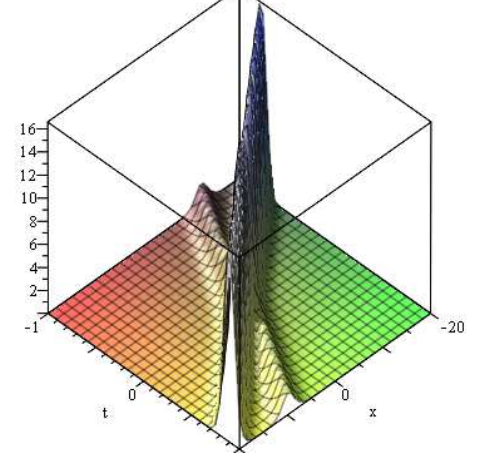}
\caption{The main approximation
$ Y_0 (x, t, \tau) $ as $\varepsilon=1$ (at the left) and $\varepsilon=0.5$ (at the right).} \label{fig: fig5}
\end{figure}

The graphs of the constructed main approximation are presented in Figure~\ref{fig: fig5}.

\section{One-phase peakon-like solution}\label{section4}

We now proceed to the problem of constructing asymptotic peakon-like solutions to equation \eqref{CHolm_vc}. Analogously to the construction of soliton-like solutions, we employ the methodology presented in \cite{SamBrandSam}, according to which the peakon-like solutions possess a structure similar to that of their soliton-like counterparts. Specifically, they consist of a regular component, representing the smooth background, and a singular component, capturing the peakon-type features of the asymptotic solution.

Although the differential equations governing the regular and singular components of the asymptotic peakon-like solutions coincide with those in the soliton-like case, a crucial distinction lies in the functional setting for the singular part. In the peakon case, solutions for the singular component are sought in different functional spaces that reflect the specific features of peakon solutions. This choice is essential, as it guarantees the characteristic properties of asymptotic peakon-like solutions: the terms forming the singular part of the asymptotics belong to appropriately selected functional spaces, ensuring the correct peakon behavior.

\subsection{Main definitions}
In the sequel we use definitions of special functional spaces. Let $ G^+ $
be a~space of infinitely differentiable functions $ f\colon {\mathbb R} \times [0; T] \times {\mathbb R_+} \mapsto {\mathbb R} $
such that for any nonnegative integers $ n $, $ p $, $ q $ and $ r $
$$
\lim_{\tau \to + \infty} \tau^n \frac{\partial ^p}{\partial x^p} \frac{\partial ^q}{\partial
 t^q} \frac{\partial ^r}{\partial \tau^r} f (x, t, \tau) = 0, \quad (x, t) \in K,
$$
uniformly with respect to $ (x, t) \in K $ in any compact set $ K \subset {\mathbb R} \times [0; T] $.

Here and below $ {\mathbb R}_+ = [0; +\infty) $.

We denote by $ {\widetilde G}^+ $ 
a space of infinitely differentiable functions $ f\colon [0; T] \times {\mathbb R_+} \mapsto {\mathbb R} $ 
such that for any nonnegative integers $ n $, $ p $ and $ q $
$$
\lim_{\tau \to + \infty} \tau^n \frac{\partial ^p}{\partial
 t^p} \frac{\partial ^q}{\partial \tau^q} f (t, \tau) = 0, \quad t \in [0; T].
$$

Analogously, let $ G^{-} $
be a space of infinitely differentiable functions $ f\colon {\mathbb R} \times [0; T] \times {\mathbb R_-} \mapsto {\mathbb R} $
such that for any nonnegative integers $ n $, $ p $, $ q $ and $ r $
$$
\lim_{\tau \to - \infty} \tau^n \frac{\partial ^p}{\partial x^p} \frac{\partial ^q}{\partial
 t^q} \frac{\partial ^r}{\partial \tau^r} f (x, t, \tau) = 0, \quad (x, t) \in K,
$$
uniformly with respect to $ (x, t) \in K $ in any compact set $ K \subset {\mathbb R} \times [0; T] $, and let $ {\widetilde G}^- $
be a space of infinitely differentiable functions $ f\colon [0; T] \times {\mathbb R_-} \mapsto {\mathbb R} $
such that for any nonnegative integers $ n $, $ p $ and $ q $
$$
\lim_{\tau \to - \infty} \tau^n \frac{\partial ^p}{\partial
 t^p} \frac{\partial ^q}{\partial \tau^q} f (t, \tau) = 0, \quad t \in [0; T].
$$
Here and below $ {\mathbb R}_- = ( -\infty; 0] $.

Thus, the functions belonging to the spaces $ G^{-} $ and $ G^+ $ are defined to the left and to the right of the point $ \tau = 0 $, respectively. The following space is obtained by identifying (or ``gluing'') functions from these two spaces along $ \tau = 0 $, thereby producing continuous functions whose asymptotic behavior at infinity is determined by the properties of elements of $ G^{-} $ and $ G^+ $. We denote by $ G^{\pm} $ 
a space of continuous functions $ f\colon {\mathbb R} \times [0; T] \times {\mathbb R} \mapsto {\mathbb R} $
such that
$$
f = \begin{cases}
 f^+(x, t, \tau), & (x,t,\tau) \in {\mathbb R} \times [0; T] \times {\mathbb R_+}, \\
 f^-(x, t, \tau), & (x,t,\tau) \in {\mathbb R} \times [0; T] \times {\mathbb R_-},
 \end{cases}
$$
where $ f^+
\in G^+ $ and $ f^-
\in G^- $, and by $ {\widetilde G}^{\pm} $ a space of continuous functions $ f\colon [0; T] \times {\mathbb R} \mapsto {\mathbb R} $
such that
$$
f = \begin{cases}
 f^+( t, \tau), & (t,\tau) \in [0; T] \times {\mathbb R_+}, \\
 f^-( t, \tau), & (t,\tau) \in [0; T] \times {\mathbb R_-},
 \end{cases}
$$
where $ f^+
\in {\widetilde G}^+ $ and $ f^-
\in {\widetilde G}^- $.

\begin{definition}[\cite{SamBrandSam}] \label{def_3}
A nontrivial function $ u = u(x,t,\varepsilon) $, where $ (x, t) \in {\mathbb R} \times[0; T] $ and $ \varepsilon $ is a small parameter, is called an asymptotic peakon-like function if for any integer $ N \ge 0 $ it can be represented as
\begin{multline}\label{1as_peakon_sol}
u(x,t,\varepsilon) = \sum\limits_{j=0}^N \varepsilon^j \left[u_j(x,t) + V_j(x,t,\tau) \right] + O\left( \varepsilon^{N+1} \right),  \quad \tau = \frac{x-\varphi(t)}{\varepsilon},
\end{multline}
where $ \varphi \in C^\infty ([0;T]) $ is a scalar function, $ u_j \in C^\infty $, $ V_j \in G^{\pm} $, $ j = 0, 1, \dots, N $.
\end{definition}

The function $ \varphi = \varphi(t) $ defines the phase curve
$$
\Gamma = \{ (x, t) \in {\mathbb R} \times [0; T] \mid x - \varphi(t) = 0\}
$$
of the peakon-like solution.

\subsection{Singular part}
The asymptotic one-phase peakon-like solution is sear\-ched in the form \eqref{1as_peakon_sol}.
As before, the regular com\-po\-nents are obtained as solutions to differential equa\-tions \eqref{reg_part_0} and \eqref{reg_part_j}, whose construction and sol\-va\-bi\-lity are discussed in Section~\ref{section2.2}.
The singular components are defined as solutions to the third-order differential sys\-tem \eqref{singular_part_0_soliton_sol} and \eqref{singular_part_1_soliton_sol} considered within appropriate func\-tio\-nal spaces.

Owing to the properties of the functions $ V_j \in G^{\pm} $, $j= 0, 1, \dots, N $, the analysis of systems \eqref{singular_part_0_soliton_sol} and \eqref{singular_part_1_soliton_sol} differs from previous cases and is carried out as follows. We begin by examining the singular terms on the phase curve $ \Gamma $. In contrast to the soliton-like case, the main term of the singular component of the asymptotic peakon-like solutions can be determined explicitly. The procedure for identifying this main term simultaneously yields the phase function~$ \varphi $. Using the structural properties of the singular terms, we subsequently construct their expansion in a neighborhood of the curve~$ \Gamma $.

We now proceed to a rigorous examination of the algorithmic procedure. The function $ v_0(t, \tau) $, introduced in \eqref{sing_term_curve}, is required to satisfy equation \eqref{sing_part_02_soliton_sol_1} and, in addition, is assumed to belong to the space $ {\widetilde G}^{\pm} $, i.e., to exhibit peakon-type characteristics. For this reason, we represent it in the form
\begin{equation} \label{CHpeakon_main_term}
v_0(t, \tau) = e^{-\alpha |\tau|}.
\end{equation}
This representation satisfies \eqref{sing_part_02_soliton_sol_1} provided that
$$
\alpha = \alpha (t) = \sqrt{\frac{b_0(\varphi (t), t)}{3}}.
$$

Simultaneously, we obtain a differential equation for the phase function $ \varphi $:
\begin{equation}\label{dif_eq_peakon}
\frac{d \varphi}{d t} = \frac{3 b_0(\varphi, t) u_0 (\varphi, t)}{3 a_0(\varphi, t) - b_0(\varphi, t)}.
\end{equation}
This differential equation is nonlinear, and its solution is, in general, defined only on a finite temporal interval, for instance, $ [0;T] $.

Thus, the main singular term on the curve $ \Gamma $ is written as \eqref{CHpeakon_main_term} under the condition $ b_0(\varphi (t), t) > 0 $, $ t\in[0; T] $, which agrees with the assumption concerning main terms of asymptotic expansions \eqref{coeff} for the coefficients equation \eqref{CHolm_vc}.
It is clear, that function \eqref{CHpeakon_main_term} belongs to the space $G^{\pm}$ and can be naturally extended as $ V_0(x, t, \tau ) = v_0(t, \tau) $.

\begin{remark} \label{rem: r_1}
The main singular term can be represented in more general form
$$
v_0(t, \tau) = A(t) e^{-\alpha(t) |\tau - \beta (t)|}
$$
with smooth functions $ A (t) $ and $ \beta (t) $. The algorithm for constructing the asymptotic peakon-like solution remains entirely unchanged. The inclusion of $ \beta(t) $ produces only a slight deformation of the curve $ \Gamma $, while $ A (t) $ serves as a scaling coefficient for wave-like solution. It is thus sufficient to consider the case $ \beta (t) = 0 $, $ A(t) = 1 $.
\end{remark}

\begin{remark} \label{rem: r_2}
In classical case of the CH equation, when $ a_0(\varphi, t) = 1 $ and $ b_0(\varphi, t) = 3 $ in \eqref{CHolm_vc}, \eqref{coeff} the main regular term is trivial, i.e., $ u_0(x, t) = 0 $; the equation for the phase function \eqref{dif_eq_peakon} becomes inapplicable and the main singular term $ v_0(t, \tau) = e^{-|\tau|} $ can depend on an arbitrary smooth function $ \varphi $ accordingly~\eqref{1as_peakon_sol}.
\end{remark}


We now turn to equations \eqref{sing_part_j2_soliton_sol_1}, which determine the higher terms $ v_j(t, \tau ) $, $ j = 1, \dots, N $, along the curve $ \Gamma $. The solution procedure is based on constructing $ v_j(t, \tau ) $ separately for the domains $ \tau > 0 $ and $ \tau <0 $. The resulting solutions are then extended and ``glued'' together so that they are continuous in $ \tau $ and belong to the space~$ {\widetilde G}^{\pm} $.

We represent components $ v_j (t, \tau) $ as
\begin{equation}\label{sol_peakon_j}
v_j (t, \tau) =
\begin{cases}
	v_j^+(t, \tau), & \tau \ge 0, \\
	v_j^-(t, \tau), & \tau < 0,
\end{cases}
\end{equation}
subject to the assumption that the functions satisfy the following conjugation conditions
\begin{equation}\label{cont_condition}
\lim_{\tau \to 0+} v_j^+(t, 0) = \lim_{\tau \to 0-}v_j^-(t, \tau).
\end{equation}

The functions in \eqref{sol_peakon_j} are solutions of corresponding differential equations
\begin{align}
& ( \varphi' - v_0 ) \frac{\partial^3 v_j^{\pm}} {\partial\tau^3} + [b_0(\varphi, t) u_0(\varphi, t) - a_0(\varphi, t) \varphi' ] \frac{\partial v_j^{\pm}}{\partial \tau} \nonumber\\
&\qquad{} + b_0(\varphi, t) \frac{\partial }{\partial \tau } ( v_0 v_j^{\pm} ) - 2 \frac{\partial }{\partial \tau } \left( \frac{\partial v_0^{\pm}}{\partial \tau} \frac{\partial v_j^{\pm}}{\partial \tau} \right)
- \frac{\partial^3 v_0^{\pm}}{\partial \tau^3} v_j = {\cal F}_j^{\pm},\label{sing_part_j2_+-}
\end{align}
where the right-side functions $ {\cal F}_j^{+} (t, \tau) $, $ {\cal F}_j^{-} (t, \tau) $ are defined recursively after the functions $
V_0, $ $ V_1 $, $ \dots $ , $ V_{j-1} $, $j= 1, \dots, N $, have been found.
In particular,
$$
{\cal F}_1^{\pm} (t, \tau) = A_1^{\pm} e^{-\alpha|\tau|} + A_2^{\pm} \tau e^{-\alpha|\tau|} + A_3^{\pm} e^{-2\alpha|\tau|} + A_4^{\pm} \tau e^{-\alpha|\tau|},
$$
where
\begin{align*}
A_1^{\pm} = 2 \alpha \alpha' \pm a_0 \alpha' - \left(\pm \alpha^2 \alpha' \pm a_1 \alpha \varphi'\right) ,
\\
A_2^{+} = - A_2^{-} = b_{0 x} u_0 + b_{0} u_{0 x} - a_{0x} \varphi',
\\
A_3^{+} = - A_3^{-} = b_1 \alpha , \quad A_4^{+} = - A_4^{-} = b_{0x} \alpha ,
\end{align*}
and the values $ a_0(x, t) $, $ a_1(x, t) $, $ b_0(x, t) $, $ b_1(x, t) $, $ u_0(x, t) $ together with their first derivatives, have been evaluated at $ x = \varphi(t) $.

After integrating equation \eqref{sing_part_j2_+-} in $ \tau $ we come to the second order inhomo\-geneous ordinary differential equations for the functions $ v_j^{-} (t, \tau) $, $ v_j^{+} (t, \tau) $ :
\begin{align}
& (\varphi' - v_0) v_{j \tau \tau }^{\pm} - v_{0 \tau } v_{j \tau }^{\pm} \nonumber\\
&  + ( - a_0(\varphi, t) \varphi'
 + b_0(\varphi, t) (v_0 + u_0 (\varphi, t) ) - v_{0 \tau \tau } )v_j^{\pm} = {\Phi}_j^{\pm},\label{ht_1}
\end{align}
where
$$
{\Phi}_j^{\pm} = {\Phi}_j^{\pm} ( t, \tau) = \int_{-\infty}^\tau {\cal F}_j^{\pm}(t, \xi) d \xi + E_j^{\pm}(t), 
$$
and $ E_j^{\pm}(t) $, $ j = 1, \dots, N $, are integration constants. For example,
\begin{align}
\Phi_1^{+} ={}& - \frac{1}{\alpha} \left(A_1^+ + A_2^+ \left( \frac{1}{\alpha} + \tau \right)\right) e^{-\alpha \tau}
\nonumber\\
& -
\frac{1}{2\alpha} \left(A_3^+ + A_4^+ \left( \frac{1}{2\alpha} + \tau\right) \right) e^{ - 2\alpha \tau} + E_1^+(t),\label{function_Phi_1+}
\\
\Phi_1^{-} ={}& \frac{1}{\alpha} \left(A_1^- + A_2^- \left(- \frac{1}{\alpha} + \tau \right)\right) e^{\alpha \tau}
\nonumber\\
& +
\frac{1}{2\alpha} \left(A_3^- + A_4^- \left(- \frac{1}{2\alpha} + \tau\right) \right) e^{ 2\alpha \tau} + E_1^-(t) .\label{function_Phi_1-}
\end{align}

Since the functions $ v_j(t, \tau ) $ inherit the properties of the singular terms $ V_j $ at the infinity, the conditions
\begin{equation}\label{function_Phi_infinity}
\lim_{\tau \to +\infty} \Phi_j^+ (t, \tau) = 0, \quad \lim_{\tau \to -\infty} \Phi_j^- (t, \tau) = 0
\end{equation}
must be satisfied. These requirements uniquely determine the choice of the right-hand side functions in \eqref{ht_1}.

Fundamental system of solutions for equation \eqref{ht_1} is given by
\begin{align*}
v_{01}^+ & = e^{-\alpha \tau}, \\
v_{02}^+ & = \frac{1}{\alpha \varphi'^3} \left(\varphi' + \frac{\varphi'^2}{2} e^{\alpha \tau} + e^{- \alpha \tau} \ln{\left| e^{\alpha \tau} - \frac{1}{\varphi'}\right|} \right) \qquad{} \\ & \mbox{for } \tau \ge 0,
\end{align*}
and
\begin{align*}
v_{01}^- & = e^{\alpha \tau}, \qquad{} \\
v_{02}^- & = - \frac{1}{\alpha \varphi'^3} \left(\varphi' + \frac{\varphi'^2}{2} e^{-\alpha \tau} + e^{\alpha \tau} \ln{\left| e^{-\alpha \tau} - \frac{1}{\varphi'}\right|} \right) \qquad{} \\ &
\mbox{for } \tau < 0. \qquad{}
\end{align*}

This system yields the general solution to \eqref{ht_1}, given by the formula \eqref{sol_peakon_j}, where $ v_j^+ $, $ v_j^- $ are written as
\begin{align}
v_j^{\pm} (t, \tau) ={}& - v_{01}^{\pm} (t, \tau) \int_{0}^{\tau} \Phi_j^{\pm}(t, \tau) v_{02}^{\pm}(t, \tau) d \tau
\nonumber\\
& + v_{02}^{\pm}(t, \tau) \int_{0}^\tau \Phi_j^{\pm}(t, \tau) v_{01}^{\pm}(t, \tau) d \tau \nonumber\\
& + c_{j1}^{\pm} v_{01}^{\pm} (t, \tau) + c_{j2}^{\pm} v_{02}^{\pm} (t, \tau),\label{sing_part_peakon_j}
\end{align}
where due to the property $v_j^{\pm} (t, \tau) \in {\widetilde G}^{\pm} $ the constants
\begin{align}
c_{j2}^{\pm} = - \lim\limits_{\tau \to \pm\infty} \int\limits_{0}^\tau \Phi_j^\pm(t, \tau) v_{01}^\pm (t, \tau) d \tau,
\end{align}
and the constants  $ c_{j1}^{\pm} $ are chosen to satisfy conjugation condition \eqref{cont_condition} at $ \tau =0 $:
$$
c_{j1}^{+} v_{01}^{+}(t, 0) + c_{j2}^{+} v_{02}^{+} (t, 0) = c_{j1}^{-} v_{01}^{-}(t, 0) + c_{j2}^{-} v_{02}^{-} (t, 0).
$$

The functions defined in \eqref{sing_part_peakon_j} enable the construction of $ v_{j} (t, \tau) $ in accordance with \eqref{sol_peakon_j}. It is evident that this function exhibits a peak at $ \tau = 0 $.

Below we demonstrate that the resulting function $ v_{j} $ belongs to the space~$ {\widetilde G}^{\pm} $.

\begin{lemma}\label{lem: lemma_6} Let $ b_0(x, t) > 0 $ for all $ (x, t) \in {\mathbb R} \times [0; T] $, the function $ \varphi $ be a solution of equation \eqref{dif_eq_peakon} and $ \varphi'(t) > 1 $ for all $ t \in [0; T] $. Then the function $ v_j $ defined by \eqref{sol_peakon_j}, \eqref{sing_part_peakon_j} belongs to the space $ {\widetilde G}^{\pm} $ and satisfies the following inequality
$$
\left| v_j^{\pm} (t, \tau) \right| \le \sum\limits_{k=0}^{j+1} C_{j k} |\tau|^{k} e^{-\alpha |\tau|},
$$
where the constants $ C_{j k} \ge 0 $ are real for all $ k = 0, 1, \dots, $ $ j+1 $.
\end{lemma}

\begin{proof}
Proving Lemma~\ref{lem: lemma_6} is straightforward and follows directly from explicit calculations applied to the solution given by formula~\eqref{sing_part_peakon_j}.
Because the reasoning for general $ j $ is completely analogous, we present the proof only for the case $ v_1 $ as the statement for other $ j $ is proved in a similar manner.

The function $ v_1 $ is defined by formula \eqref{sing_part_peakon_j} with $ j = 1 $.
The correspondent right-side functions satisfy conditions \eqref{function_Phi_infinity} which yield $ E_1^+ (t) = E_1^- (t) = 0 $ in \eqref{function_Phi_1+}, \eqref{function_Phi_1-}. Then the right-side functions $ \Phi_1^\pm $ in \eqref{sing_part_peakon_j} obviously satisfy the inequalities
\begin{equation}\label{cond_Phi_1_pm}
\left| \Phi_1^\pm (t, \tau) \right| \le \left( A_1 |\tau| + A_2 \right) e^{-\alpha |\tau|}
\end{equation}
for some non-negative $ A_1 $, $ A_2 $ .

Taking into account exact form of the fundamental system of solutions for equation \eqref{ht_1} and assumption $ \varphi'(t) > 1 $, $ t \in [0; T] $, we can write inequalities
\begin{align}
\left|v_{01}^{\pm} (t, \tau)\right| &\le e^{-\alpha |\tau|}, \nonumber\\
\left|v_{02}^{\pm} (t, \tau)\right| & \le B_1 + B_2 e^{\alpha |\tau|}
+ B_3 e^{-\alpha |\tau|} \ln\left| e^{-\alpha |\tau|} - \frac{1}{\varphi'}\right|,\label{cond_fund_sol}
\end{align}
for some non-negative constants $ B_1 $, $ B_2 $, $ B_3 $.

From formula \eqref{sing_part_peakon_j} in virtue of \eqref{cond_Phi_1_pm}, \eqref{cond_fund_sol}, we deduce
\begin{align*}
\left| v_1^{\pm} (t, \tau) \right| & \le \left| v_{01}^{\pm} (t, \tau)  \right| \int_{0}^{\tau} \left| \Phi_1^{\pm}(t, \xi) \right| \left| v_{02}^{\pm}(t, \xi) \right| d \xi \\
&  + \left| v_{02}^{\pm}(t, \tau) \right| \int_{\tau}^{+\infty} \left|\Phi_1^{\pm}(t, \xi) \right| \left| v_{01}^{\pm}(t, \xi) \right| d \xi \\
&  + |c_{11}^{\pm} | \, \left| v_{01}^{\pm} (t, \tau) \right| \\
& \le e^{-\alpha |\tau|} \int_0^\tau \left( A_1 |\xi| + A_2 \right) e^{-\alpha |\xi|} \\
& \times \left(B_1 + B_2 e^{\alpha |\xi|} + B_3 e^{-\alpha |\xi|} \ln\left| e^{-\alpha |\xi|} - \frac{1}{\varphi'}\right| \right) d \xi
\\
& + \left(B_1 + B_2 e^{\alpha |\tau|} + B_3 e^{-\alpha |\tau|} \ln\left| e^{-\alpha |\tau|} - \frac{1}{\varphi'}\right| \right) \\
& \times \int_{\tau}^{+\infty}
\left( A_1 |\xi| + A_2 \right) e^{-2\alpha |\xi|} d \xi + |c_{11}^{\pm} | \, e^{-\alpha|\tau|}
\\
\le & \left( C_1 |\tau|^2 + C_2 |\tau| + C_3 \right) e^{-\alpha |\tau|}
\end{align*}
for suitable constants $ C_1 $, $ C_2 $ and $ C_3$.
\end{proof}

Owing to Lemma~\ref{lem: lemma_6}, the function $ v_j $ can be extended as $ V_j(x, t, \tau ) = v_j(t, \tau) $. Summarizing these results, we obtain the following statement.

\begin{theorem} \label{thm: th_5}
Let the following conditions be assumed:
\begin{enumerate}
\item[1.] The functions $ a_j $, $ b_j \in C^{\infty} ({\mathbb R}\times [0; T])$, $ j = 0, 1, \dots,$ $ N $, and additionally
$ a_0(x, t) b_0(x, t) \not= 0 $, $ b_0(x, t) > 0 $ for all $ (x, t) \in {\mathbb R}\times [0; T] $;
\item[2.] The function $ \varphi $ is a solution of differential equation \eqref{dif_eq_peakon} and the inequality $ \varphi'(t) > 1 $ holds for all $ t \in [0;T] $.
\end{enumerate}

Then the function
\begin{equation}\label{as_sol_peakon}
Y_N (x, t, \varepsilon) = \sum\limits_{j=0}^N \varepsilon^j \left[u_j(x, t) + V_j (x, t, \tau)\right], \quad \displaystyle \tau =
\frac{x-\varphi(t)}{\varepsilon},
\end{equation}
is the $ N $-th asymptotic approximation of the peakon-like solution of equation \eqref{CHolm_vc} and satisfies this equation on the set
\begin{align*}
\{(x, t)\in {\mathbb R}\times [0; T] \mid x-\varphi(t) > 0\} \\
\cup \{(x, t)\in {\mathbb R}\times [0; T] \mid x-\varphi(t) < 0 \}
\end{align*}
with an asymptotic accuracy $ O(\varepsilon^N) $.
In addition, as $ \tau \to \pm\infty $ function \eqref{as_sol_peakon} satisfies \eqref{CHolm_vc} with an asymptotic accuracy $ O(\varepsilon^{N+1}) $.
\end{theorem}

\begin{proof}
The idea behind the proof of Theorem~\ref{thm: th_5} is analogous to that of Theorem~5 in~\cite{SamBrandSam}, which concerns the asymptotic one-phase peakon-like solution of the vcmCH equation. Therefore, we only outline its main idea.

In constructing the asymptotic solution, the general principles of asymptotic analysis are applied, inherently ensuring that the solution being developed meets the intended asymptotic accuracy. A rigorous justification is then required to confirm that the constructed solution indeed achieves the desired level of precision.

In the proof, the structure of the constructed solution is carefully taken into account, and the solution is effectively analyzed on two distinct domains corresponding to $\tau > 0$ and $\tau < 0$. Although the derivatives of the singular terms with respect to the variable $\tau$ are discontinuous, this does not complicate the fulfillment of the asymptotic estimates. These estimates are obtained by carefully accounting for the specific properties of functions belonging to the spaces $G^{\pm}$.

This approach relies on the analysis of the corresponding residual functions~\cite{Sam_MMC_2021_2} for each domain mentioned above. Since this analysis is purely technical and can be rather tedious, we do not present it in detail here.
\end{proof}

\subsection{Example 3}

We consider the construction of the first two terms of the asymptotic one-phase peakon-like solution \eqref{1as_peakon_sol} for the vcCH equation of the form
\begin{multline} \label{example_3}
\frac{5}{4} \bigl(x^2 + 4\bigr) u_t - \varepsilon^2 u_{xxt} + 6 \bigl(t^2 + 1\bigr) u u_x - 2 \varepsilon^2 u_x u_{xx} - \varepsilon^2 uu_{xxx} = 0.
\end{multline}

This equation has a form \eqref{CHolm_vc}, where
$$
a(x, t, \varepsilon) = \frac{5}{4} \bigl(x^2 + 4\bigr), \quad
b(x, t, \varepsilon) = 6 \bigl(t^2 + 1\bigr).
$$

Taking into account the coefficients of this equation, the equation for the phase function can be readily obtained by specifying~\eqref{dif_eq_peakon}. This equation can be further simplified by taking the regular part in the form $u_0(x, t) = u_1(x, t) = 1$, which satisfies the equations for the regular terms \eqref{reg_part_0} and \eqref{reg_part_j}. Under this assumption, equation \eqref{dif_eq_peakon} for the phase function $\varphi$ has a global solution $\varphi = 2t$ with $ t \in {\mathbb R} $.

The main singular term along curve $ x = 2 t $ is exactly determined from \eqref{CHpeakon_main_term} with $ \alpha = \sqrt{2 (t^2 + 1)} $ and is given by
\begin{equation}\label{example_sol_main-term_peakon}
 v_0(t, \tau) = e^{-\alpha |\tau|} = \exp{\left(-\sqrt{2(t^2 + 1)} |\tau|\right)}.
\end{equation}

This function belongs to the space $ {\widetilde G}^{\pm} $ and can be extended in a standard way $ V_0(x, t, \tau) = v_0(t, \tau) $. So, the main term of the asymptotic one-phase peakon-like solution is given as
$$
Y_0(x, t, \varepsilon) = 1 + \exp{\left(-\sqrt{2(t^2 + 1)} \left|\frac{x - 2t}{\varepsilon} \right| \right)},\quad (x, t)\in {\mathbb R}^2.
$$

To find the next singular term, we calculate the functions in \eqref{function_Phi_1+}, \eqref{function_Phi_1-}:
\begin{align*}
\Phi_1^+(t, \tau) & = \left( \frac{10 t \tau}{\alpha} + \frac{10 t}{\alpha^2} - \frac{4 t}{\alpha} - 3t \right) e^{-\alpha \tau } ,
\\
\Phi_1^-(t, \tau) & = \left(\frac{10 t \tau}{\alpha} - \frac{10 t}{\alpha^2} + \frac{4 t}{\alpha} - 3t \right) e^{\alpha \tau }
\end{align*}
by using of which we obtain the first singular term $ v_1 $ in the following form:
\begin{align*}
v_1 (t, \tau) ={}& \frac{5 t}{4\alpha^3} \tau e^{-2 \alpha |\tau|} - \frac{3 t}{8\alpha^2} e^{-2 \alpha |\tau|} \\
& + \left(\frac{25 t}{8\alpha^4} - \frac{t}{2 \alpha^3} \right) e^{-2 \alpha |\tau|} \operatorname{sign}{(\tau)}\\
& - \frac{5t}{4 \alpha^2} \tau^2 e^{ -\alpha|\tau|} \operatorname{sign}{(\tau)} - \left( \frac{15 t}{4 \alpha^{3}} - \frac{t}{\alpha^2} \right) \tau e^{-\alpha|\tau|} \\
& + \frac{3 t}{4 \alpha} \tau e^{-\alpha |\tau|} \operatorname{sign}{(\tau)}
+ \left(\frac{15 t}{8 \alpha^{2}} - \frac{3 t \ln 2}{16 \alpha^2} \right) e^{-\alpha |\tau|} \\
& -
 \left(\! \frac{25 t}{8 \alpha^4} - \frac{t}{2 \alpha^3} + \frac{5t \ln 2}{4 \alpha^8} -  \frac{t \ln 2}{4 \alpha^3} \!\right)\! e^{-\alpha |\tau|} \operatorname{sign}{(\tau)} \\
 & - \frac{1}{8 \alpha} e^{-\alpha |\tau|} \operatorname{sign}{(\tau)} \\
 & \times \int_{0}^{\tau} \left( \frac{10 t \tau}{\alpha} - 3 t + \left(\frac{10 t}{\alpha^2} - \frac{4 t}{\alpha} \right) \right. \\
 & \left. \operatorname{sign}{(\tau)} \right) e^{-2\alpha |\tau|} \ln \left|e^{\alpha |\tau|} - \frac{1}{2} \right| d \tau \\
& - \left(\frac{5 t \tau}{8\alpha^3} - \frac{3t}{16\alpha^2} \right) e^{-3 \alpha |\tau|} \ln \left|e^{\alpha |\tau|} - \frac{1}{2} \right| \\
& - \left( \frac{5 t}{4 \alpha^8} - \frac{t}{4 \alpha^{3}} \right)e^{-3 \alpha |\tau|} \operatorname{sign}{(\tau)} \ln \left|e^{\alpha |\tau|} - \frac{1}{2} \right|,
\end{align*}
where
$$
\operatorname{sign}{(\tau)} = \begin{cases}
 \hphantom{-}1, & \tau \ge 0, \\
 -1,& \tau < 0, \\
 \end{cases}
\quad \text{and}\quad \tau = \frac{x-2t}{\varepsilon} .$$

It is easy to see that the function $ v_1 \in {\widetilde G}^{\pm} $ and it can be extended in required manner as
$$
V_1(x, t, \tau) = v_1(t, \tau).
$$

Thus, the first asymptotic approximation of the asymptotic peakon-like solution is written as
\begin{align}\label{as_sol_1_peakon}
Y_1 (x, t, \tau) = u_0(x, t) + V_0 (x, t, \tau) + \varepsilon (u_1(x, t) + V_1 (x, t, \tau) ).
\end{align}

The graphs of function \eqref{as_sol_1_peakon} are presented in Figures~\ref{fig: 5}--\ref{fig: 7}.

\begin{figure}[ht]
\centering
\includegraphics[scale=0.51]{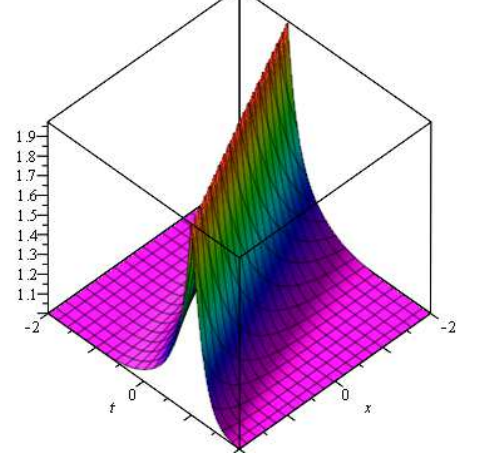} \quad
\includegraphics[scale=0.51]{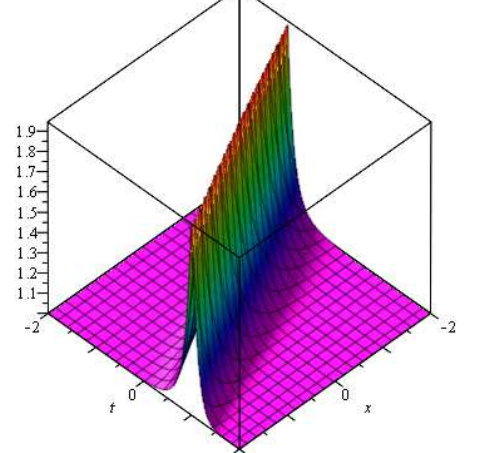}
\caption{The main term 
$ Y_0 (x, t, \tau) = u_0(x, t) + V_0 (x, t, \tau)$ as $\varepsilon=1$ (at the left) and $\varepsilon=0.5$ (at the right).} \label{fig: 5}
\end{figure}

\begin{figure}[ht]
\centering
\includegraphics[scale=0.51]{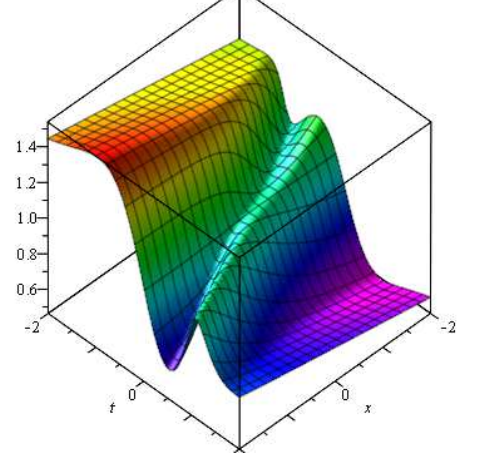} \quad
\includegraphics[scale=0.51]{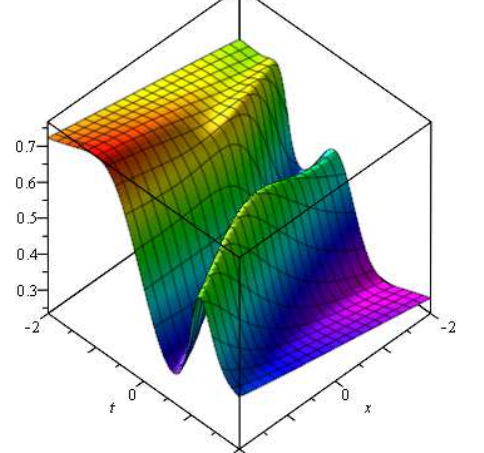}
\caption{The first term 
$ \varepsilon(u_1(x, t) + V_1 (x, t, \tau) ) $ as $\varepsilon=1$ (at the left) and $\varepsilon=0.5$ (at the right).} \label{fig: 6}
\end{figure}

\begin{figure}[ht]
\centering
\includegraphics[scale=0.51]{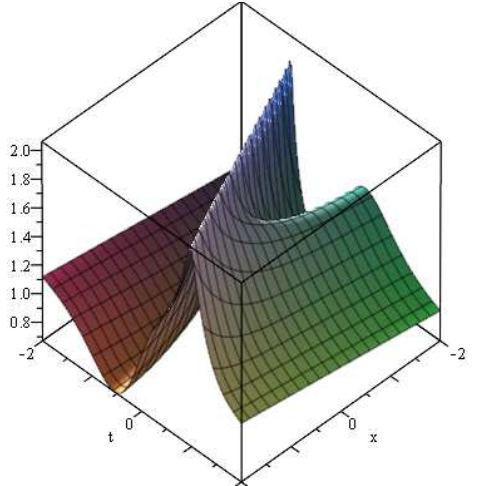} 
\includegraphics[scale=0.51]{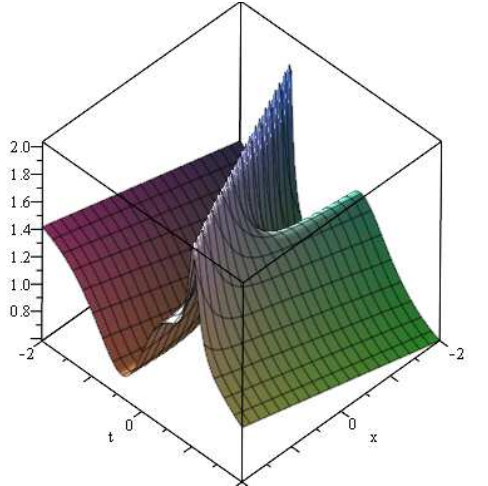}
\caption{The asymptotic solution 
$ Y_1 (x, t, \tau)$ as $\varepsilon=1$ (at the left) and $\varepsilon=0.5$ (at the right).} \label{fig: 7}
\end{figure}

\section{Two-phase peakon-like solutions}\label{section5}

Taking into account the integrability nature of the Camassa–Holm equation possessing multi-phase peakon solutions, it is natural to consider the problem of constructing asymptotic two-phase peakon-like solutions to the vcCH equation~\eqref{CHolm_vc}.
Such a solution is sought in the form of a sum of the regular and singular parts of the asymptotic expansion, whose terms are specified by equations~\eqref{reg_part_0},~\eqref{reg_part_j}, and~\eqref{sing_term_0_2phase}. As in the problems considered above, the key difficulty lies in determining the singular part of the asymptotics, since it is precisely this part that reflects the qualitative properties of the solution, and its construction requires finding particular solutions of certain partial differential equations within appropriate functional spaces.

Furthermore, similar to the case of the two-phase so\-li\-ton-like solutions, the complexity of the problem increases compared to the case of the one-phase peakon-like solutions. The reasons for this complication are analogous to those encountered in the construction of two-phase soliton-like solutions (see the introduction to Section~\ref{section3}). However, unlike the soliton case, here one can obtain a simpler representation by exploiting the structure of the two-peakon solution of the classical CH equation.
Nevertheless, for higher-order singular terms, it appears impossible to overcome the difficulties associated with their construction.

\subsection{Main definitions and concepts}

Let $ G_2^{\pm} $ be a space of continuous functions
$$
f = f(x, t, \tau_1, \tau_2) \in C^\infty
({\mathbb{R}} \times [0; T] \times {\mathbb R} \times {\mathbb R}), $$
for which there exist functions $ f_1^{\pm} = f_1^{\pm} (x,t, \tau_2) $, $ f_2^{\pm} = f_2^{\pm} (x,t, \tau_1) \in G^{\pm} $, such that for all non-negative integers $ n $, $ p $, $ q $, $ \alpha $, $\beta $ uniformly with respect to $ (x,t) \in K $ on any compact set $ K \subset {\mathbb{R}} \times [0; T] $ we have the relations
$$
\lim_{\tau_k \to \pm \infty} \tau_k^{ n} \frac {\partial^{ p}} {\partial x^p}\frac{\partial^{ q}}{\partial t^q}
\frac{\partial^{ \alpha}} {\partial \tau_1^{\alpha}} \frac{\partial^{ \beta}} {\partial \tau_2^{\beta}} \left(
f - f_k^{\pm} \right) = 0, \quad k =1, 2.
$$

Analogously to the space $ {\widetilde G}_2 $, we define $ {\widetilde G}_2^{\pm} \subset G_2^{\pm} $ as a space of continuous functions $ f = f(t, \tau_1, \tau_2) \in C^\infty ([0; T] \times {\mathbb R} \times {\mathbb R}) $, for which there exist functions  $ f_1^{\pm} = f_1^{\pm} (t, \tau_2) $,
$ f_2^{\pm} = f_2^{\pm} (t, \tau_1) \in {\widetilde G}^{\pm} $, such that for all non-negative integers $ n $, $ q $, $ \alpha $, $\beta $ uniformly in $ t \in [0; T] $ we have the relations
$$
\lim_{\tau_k \to \pm \infty} \tau_k^{ n} \frac{\partial^{ q}}{\partial t^q} \frac{\partial^{ \alpha}} {\partial \tau_1^{\alpha}} \frac{\partial^{ \beta}} {\partial \tau_2^{\beta}} \left(
f - f_k^{\pm} \right) = 0, \quad k =1, 2.
$$

\begin{definition} A nontrivial function $ u = u(x, t, \varepsilon) $, whe\-re $ \varepsilon $ is a small parameter, is called an asymptotic two-phase peakon-like function if for any integer $ N \ge 0 $ it can be represented as
	\begin{equation}\label{2as_sol_peakon}
		u(x, t, \varepsilon) = \sum\limits_{j=0}^N \varepsilon^j \left[u_j(x, t) + V_j (x, t, \tau_1, \tau_2)\right] + O\bigl(\varepsilon^{N+1}\bigr),
	\end{equation}
where
$$
\tau_1 = \frac{x - \varphi_1(t)}{\varepsilon}, \quad \tau_2 = \frac{x - \varphi_2(t)}{\varepsilon},
$$
$ V_j \in G_2^{\pm} $, $ j = 0 , 1, \dots, N $, and the phase functions $ \varphi_1 $, $ \varphi_2 \in C^{\infty} ([0;T]) $ satisfy condition $ \varphi_1(0) = \varphi_2(0) $.
\end{definition}

\begin{remark}\label{rem: r_3}
In view of the properties of the two-phase peakon-like solution, $ \varphi_1 (t) \not= \varphi_2(t) $ for all $ t \in (0; T] $. Moreover, we assume the condition $ \varphi_1^{\prime}(t) \not= \varphi_2^{\prime}(t) $ for all $ t \in [0; T] $.
\end{remark}

Formula \eqref{2as_sol_peakon} plays a key role in determining the components of the required solution. As in the soliton case, it is precisely the structure of the singular terms that allows us to derive differential equations for the regular terms of the form~\eqref{reg_part_0} and~\eqref{reg_part_j}, as well as the corresponding equations for the singular terms.

The main singular term of the asymptotic two-phase peakon-like solution is sought as a solution of differential equation~\eqref{main_term2phase_sol} in the space $ G_2^{\pm} $.
The equations for the higher singular terms are more cumbersome than those for $ V_0 $; therefore, they are not presented here.
In the sequel we consider the case of trivial background function, i.e., $ u_k(x, t) = 0 $ for all $ k \ge 0 $.

Let move on constructing main singular term.

\subsection{Main singular term}

Consider equation \eqref{main_term2phase_sol} with $ u_0(x, t) = 0 $. It is a nonlinear partial differential equation, the exact integrability of which appears to be impossible. Therefore, we impose additional conditions of the form $ a_0(\varphi_1, t) = a_0(\varphi_2, t) = 1$, $ b_0(\varphi_1, t) = b_0(\varphi_2, t) = 3 $, which allow us to construct a two-phase peakon-like solution.

Denote $ c_1 = \varphi_1' $, $ c_2 = \varphi_2' $. In view of Remark~\ref{rem: r_3}, we suppose inequality $ c_1 - c_2 > 0 $.

By introducing new variables $\xi $, $ \eta $ via
$$
\tau_k = \xi - \varphi_k'(t) \eta, \quad k = 1, 2,
$$
equation \eqref{main_term2phase_sol} transforms into \eqref{main_term_two_phase_n-var}, which is in fact the Camassa--Holm equation. This equation admits a two-peakon solution of the form \cite{Lundmark_2019}:
\begin{align}
V_0(\xi, \eta) ={}& \frac{c_1 e^{- c_1 \eta} + c_2 e^{- c_2\eta} }{e^{- c_1 \eta} + e^{- c_2 \eta}}
 \exp{\left( - \left| \xi + \ln \frac{c_1 e^{- c_1 \eta} + c_2 e^{- c_2\eta}}{c_1 - c_2} \right| \right) } \nonumber\\
& +\frac{c_1 e^{c_1 \eta} + c_2 e^{c_2\eta} }{e^{ c_1 \eta} + e^{c_2 \eta}}  \exp{\left( - \left| \xi - \ln \frac{c_1 e^{c_1 \eta} + c_2 e^{c_2\eta}}{c_1 - c_2} \right|\right)}.\label{as_sol_2_peakon}
\end{align}
By the inverse transformation of variables
\begin{equation}\label{variables_peakon}
\xi = \frac{c_2 \tau_1 - c_1 \tau_2}{c_2 - c_1}, \quad \eta = \frac{\tau_1 - \tau_2}{c_2 - c_1}
\end{equation}
we obtain the main singular term $ V_0 (t, \tau_1, \tau_2) $ in exact form.

The following properties of resulting function $ V_0 $:
\begin{align*}
\lim_{\tau_1 \to + \infty} V_0 (t, \tau_1, \tau_2) & = c_2 \exp{\left(- \left| \tau_2 - \ln \frac{c_2 }{c_1 - c_2} \right|\right)},
\\
\lim_{\tau_1 \to - \infty} V_0 (t, \tau_1, \tau_2) & = c_1 \exp{\left(- \left| \tau_2 + \ln \frac{c_2 }{c_1 - c_2} \right|\right)},
\\
\lim_{\tau_2 \to + \infty} V_0 (t, \tau_1, \tau_2) & = c_2 \exp{\left(- \left| \tau_1 - \ln \frac{c_1}{c_1 - c_2} \right|\right)},
\\
\lim_{\tau_2\to - \infty} V_0 (t, \tau_1, \tau_2) & = c_1 \exp{\left(- \left| \tau_1 + \ln \frac{c_1 }{c_1 - c_2} \right|\right)}
\end{align*}
imply that $ V_0(t, \tau_1, \tau_2) $ belongs to the space $ {\widetilde G}^{\pm}_2 $.

\begin{remark}\label{rem: r_4}
In limiting case where one of the phase variable tends to $ +\infty $ or $ -\infty $, the value of $ V_0(t, \tau_1, \tau_2) $ approaches the function that represents the main singular term of the asymptotic one-phase peakon-like solution.
\end{remark}

\begin{theorem} \label{thm: th_6}
Let the following conditions be assumed:
\begin{enumerate}
\item[1.] The functions $ a_0 $, $ b_0 \in C^{(1)} ({\mathbb R}\times [0; T])$, $ j = 0, 1, \dots, $ $ N $, $ a_0(x, t) b_0(x, t) \not= 0 $ for all $ (x, t) \in {\mathbb R}\times [0; T] $;
\item[2.] The phase functions $ \varphi_1 $, $ \varphi_2 \in C^{(1)} ([0; T]) $ are such that condition $ \varphi_1(0) = \varphi_2(0) $ holds and
\begin{align*}
& \varphi'_1(t) > \varphi'_2(t), \\
& a_0(\varphi_1(t), t) = a_0(\varphi_2(t), t) = 1 , \\
& b_0(\varphi_1(t), t) = b_0(\varphi_2(t), t) = 3
\end{align*}
for all $ t \in [0; T] $.
\end{enumerate}

Then function
$$
Y_0(x, t, \varepsilon) = V_0(t, \tau_1, \tau_2),\, \tau_1 = \frac{x-\varphi_1(t)}{\varepsilon}, \, \tau_2 = \frac{x-\varphi_2(t)}{\varepsilon},
$$
where $ V_0 $ is given via \eqref{as_sol_2_peakon}, \eqref{variables_peakon}, is the main asymptotic approximation of the two-phase peakon-like solution of equation \eqref{CHolm_vc}.
\end{theorem}

\begin{proof} This statement is analog of Theorem~\ref{thm: th_4} and that is why it is given without additional reasoning.
\end{proof}

\subsection{Example 4}

To demonstrate the result of applying the algorithm within the framework of Theorem~\ref{thm: th_6}, we consider the vcCH equation \eqref{CHolm_vc} with coefficients given as~\cite{Popovych_2012}
$$
a_0(x, t) = \exp{\bigl( x^2 - 62 xt + 600 t^2\bigr)} , \quad b_0(x, t) = 3 + \varepsilon t^2,
$$
which leads to the following equation:
\begin{align}
& \exp{\bigl(x^2 - 62 xt + 600 t^2\bigr)} u_t - \varepsilon^2 u_{txx} + \bigl(3 + \varepsilon t^2\bigr) uu_x\nonumber\\
& \qquad - 2 \varepsilon^2 u_x u_{xx} - \varepsilon^2 uu_{xxx} = 0.\label{ex_2_peakon}
\end{align}

We choose $ \varphi_1(t) = 50 t $ and $ \varphi_2(t) = 12 t $, which ensures that the conditions of Theorem~\ref{thm: th_6} are satisfied for all $ t $.

After reducing \eqref{ex_2_peakon} to the CH equation, accordingly to algorithm of constructing the main term of the singular part of the two-phase peakon-like solution we obtain
\begin{align*}
V_0(t, \tau_1, \tau_2) = \frac{50 e^{\alpha_1(\tau_1 - \tau_2)} + 12 e^{\alpha_2 (\tau_1 - \tau_2)}}{e^{\alpha_1 (\tau_1 - \tau_2)} + e^{\alpha_2 (\tau_1 - \tau_2)}} \\
 \times \exp \left( - \left| \frac{25 \tau_2 - 6 \tau_1}{19} + \ln \frac{25 e^{\alpha_1(\tau_1 - \tau_2)} + 6 e^{ \alpha_2 (\tau_1 - \tau_2)}}{19} \right| \right)
\\
 + \frac{50 e^{- \alpha_1(\tau_1 - \tau_2)} + 12 e^{- \alpha_2 (\tau_1 - \tau_2)}}{e^{-\alpha_1 (\tau_1 - \tau_2)} + e^{-\alpha_2 (\tau_1 - \tau_2)}} \\
 \times \exp\left( - \left| \frac{25 \tau_2 - 6 \tau_1}{19}  \right. \right.
 \left. \left.
 - \ln \frac{25 e^{-\alpha_1 (\tau_1 - \tau_2 )} + 6 e^{- \alpha_2 (\tau_1 - \tau_2)}}{19} \right| \right) ,
\end{align*}
where $ \alpha_1 = {25} / {19}$, $ \alpha_2 = {6} / {19}$.

It is obviously, that
\begin{align*}
\lim_{\tau_1 \to + \infty} V_0 (t, \tau_1, \tau_2) & = 12 \exp{\left(- \left| \tau_2 - \ln \frac{6}{19} \right|\right)},
\\
\lim_{\tau_1 \to - \infty} V_0 (t, \tau_1, \tau_2) & = 50 \exp{\left(- \left| \tau_2 + \ln \frac{6}{19} \right|\right)},
\\
\lim_{\tau_2 \to + \infty} V_0 (t, \tau_1, \tau_2) & = 12 \exp{\left(- \left| \tau_1 - \ln \frac{25}{19} \right|\right)},
\\
\lim_{\tau_2\to - \infty} V_0 (t, \tau_1, \tau_2) & = 50 \exp{\left(- \left| \tau_1 + \ln \frac{25}{19} \right|\right)}.
\end{align*}

The main term of the asymptotic two-phase soliton-like solution is given as
\begin{align*}
Y_0(x, t, \varepsilon) = \frac{50 e^{- 50 t / \varepsilon} + 12 e^{ -12 t / \varepsilon} }{e^{ -50 t/\varepsilon} + e^{-12 t/ \varepsilon }} \\
\times \exp{\left( - \left| \frac{x}{\varepsilon} + \ln \frac{25 e^{- 50 t / \varepsilon} + 6 e^{- 12 t / \varepsilon}}{19} \right| \right) } \\
+ \frac{50 e^{50 t / \varepsilon} + 12 e^{ 12 t / \varepsilon} }{e^{ 50 t / \varepsilon} + e^{12 t / \varepsilon}} \exp{\left( - \left| \frac{x}{\varepsilon} - \ln \frac{25 e^{50 t/\varepsilon} + 6 e^{12 t / \varepsilon}}{19} \right| \right)}.
\end{align*}

Its graphs are represented on Figure~\ref{fig: 8}.

\begin{figure}[ht]
\centering
\includegraphics[scale=0.52]{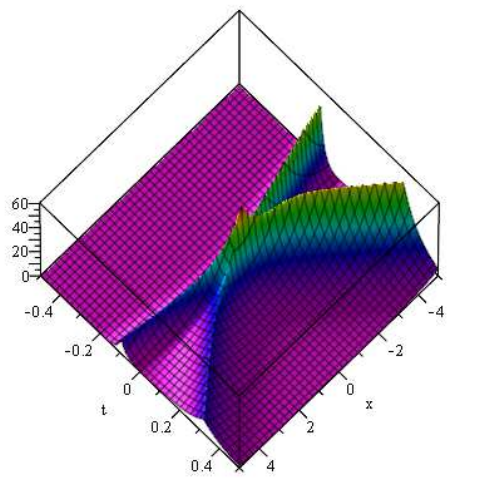} 
\includegraphics[scale=0.52]{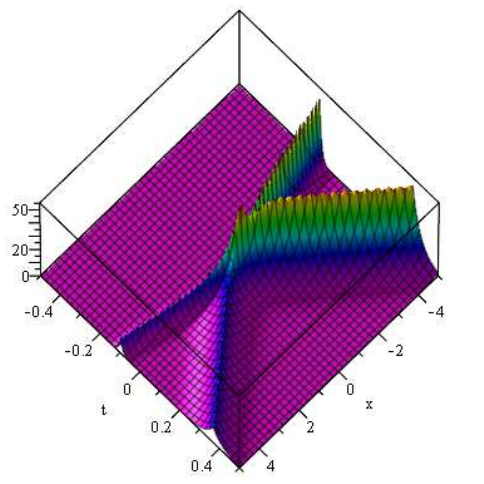}
\caption{The main term 
$ Y_0 (x, t, \tau) = V_0 (x, t, \tau)$ as $\varepsilon=1$ (at the left) and $\varepsilon=0.5$ (at the right).} \label{fig: 8}
\end{figure}

\section{Discussions and conclusions}\label{section6}

This study addresses the Camassa--Holm equation with variable coefficients under a small dispersion. By introducing variable coefficients, this model constitutes a direct generalization of the classical CH equation \cite{CamassaHolm}, which is well known  as a bi-Hamiltonian dynamical system \cite{BPS} admitting both soliton and peakon solutions \cite{Focas}.

The presence of variable coefficients in the system generally complicates finding its solutions and alters certain properties characteristic of integrable systems. Under fair\-ly general assumptions, and considering the smallness of the dispersive terms, some of these properties can be interpreted as small deformations of the corresponding geometric and physical structures of the classical CH equation with constant coefficients.

In this paper, we focus on the mathematical description of particular solutions of the vcCH equation that exhibit properties analogous, in a certain sense, to those of classical soliton and peakon solutions. The associated mathematical model contains a small parameter at the higher order derivatives and is recognized in the literature as a singularly perturbed system, which allows the application of the idea of the nonlinear WKB method \cite{Miura} for its constructive analysis.

For the vcCH system, we consider the construction of soliton- and peakon-like solutions in the form of asymptotic expansions, including both one-phase and two-phase cases. The methodology for constructing these solutions is based on approaches previously validated for the vcmCH equation \cite{SamBrandSam}. However, the vcCH system presents substantial differences and is considerably more complex.

The searched solution is represented as a sum of two asymptotic components: a regular part, which serves as a background function common to both soliton-like and pea\-kon-like solutions, and a singular component, which captures the distinctive features of these solutions.

The determination of the terms of the regular part of the asymptotic solution does not present significant difficulties, since it is connected only with the possibility of integrating certain quasilinear and linear first-order partial differential equations by means of the method of characteristics.

Searching the singular terms of the asymptotics encounters substantial difficulties arising from the fact that the main singular term of the asymptotics, as in the case of the CH equation with constant coefficients, is defined implicitly.

In the one-phase case this is not an insurmountable obstacle, since it is possible to study the solvability of the equations for the higher order singular terms, which opens the way to constructing an asymptotic soliton-like solution with arbitrary accuracy with respect to a small parameter. However, in the two-phase case even the initial step -- the construction of a rapidly decaying solution of the equation for the main singular term -- cannot, in general, be carried out using the methodology of \cite{SamBrandSam}, which should be regarded as a certain limitation of its applicability.

Here, apparently insurmountable technical obstacles arise due to the considerable complexity of integrating the differential equation for the main singular term \eqref{sing_term_0_2phase}, which contains three nonlinear terms and two independent variables. Taking into account that the classical CH equation \eqref{Camassa_Holm} does not possess the property of gauge equivalence (property of insensitivity to the values of the coefficients), and that the coefficients of its nonlinear terms play a key role in its integrability, one must apparently conclude that the main singular term, as a particular solution of a partial differential equation with time dependent coefficients -- even when the time variable in equation \eqref{sing_term_0_2phase} is treated as a parameter -- cannot be constructed in the case of arbitrary coefficients in equation \eqref{CHolm_vc}. This leads to the necessity of imposing certain restrictions on these coefficients.

This situation differs somewhat from that of the Kor\-te\-weg--de Vries equation with variable coefficients, for which asymptotic two-phase \cite{Sam_2008} and multi-phase \cite{Sam_2012_1, Sam_2012_2} so\-li\-ton-like solutions have been constructed and estimates of their asymptotic accuracy have been established. This was made possible precisely due to the property of gauge equivalence of the classical KdV equation, that is, its ``insensitivity'' to the exact values of its coefficients.

Another notable aspect of the one-phase case concerns the determination of the phase function. For the vcmCH equation, this function is determined as a solution to a nonlinear ordinary differential equation derived from an orthogonality condition. In case of the vcCH system, the only constraint on the phase function is provided through certain algebraic inequalities involving the main terms of the expansions for the original equation’s coefficients, the regular component of the asymptotic solution, and the derivative of the phase function. These relations also define, in particular, the time interval over which the corresponding asymptotic expansions for soliton-like solutions are valid.

Important results of this study concern the construction of asymptotic pea\-kon-like solutions. First, the terms of the singular part of the asymptotic solutions were constructed separately on two intervals: to the left and to the right of the point where the peakon solution loses smoothness. The resulting functions were then ``glued together'' in such a way as to satisfy the conditions of continuity and piecewise smoothness.

It should be noted that the main term of the constructed asymptotic peakon-like solution is determined explicitly, as in the case of the vcmCH equation \cite{SamBrandSam}. This is due to the simpler form of the peakon solution compared to the soliton solution. As in the soliton case, the solvability of the differential equations for the higher-order asymptotic terms of the peakon-like solution has been established in appropriate functional spaces, which makes it possible to obtain the corresponding asymptotic expansion with an arbitrary number of terms.

As in the previously considered cases, theorems on the asymptotic accuracy of the constructed asymptotic solutions have been proved.

Each of the considered cases is illustrated by nontrivial examples for which, in accordance with the obtained general results, approximate solutions are derived in explicit form and their graphs are presented.

\section*{CRediT authorship contribution statement}

\noindent
Yuliia Samoilenko: Conceptualisation, Methodology, Validation, Formal analysis, Investigation, Writing original draft, Writing review \& editing, Visualization.

\medskip

\noindent
Valerii Samoilenko: Conceptualisation, Methodology,
Validation, Formal analysis, Investigation, Writing original draft, Writing review \& editing.

\section*{Declaration of competing interest}
The authors declare that they have no known competing financial interests or personal relationships that could
have appeared to influence the work reported in this paper.

\section*{Data availability}
No data was used for the research described in the article.

\section*{Acknowledgement}

The first author was partially supported by the ANR through projects No. ANR--23--PAUK--0038--01 and ANR--25--PAUK--0038--01. The second author was partially supported by a grant from the Simons Foundation International (SFI--PD--Ukraine--00014586, V.H.S.).

The authors are deeply grateful to Professor Lorenzo Brandolese (Institut Camille Jordan, Universit\'e Claude Bernard Lyon 1, Villeurbanne, France) for drawing their attention to the Camassa–Holm equation, as well as for his constant support, valuable advice throughout the research, and insightful comments on the results presented in this paper.

The authors also sincerely thank their colleagues, Professor Anatoliy Pry\-kar\-pat\-skiy (Department of Computer Science and Telecommunication, Cracow University of Te\-ch\-nology, Kraków, Poland, and Department of Ad\-van\-ced Mathematics, Lviv Polytechnic National University, Lviv, Ukraine) and Professor Vyacheslav Boyko (Institute of Ma\-the\-matics of NAS of Ukraine), for helpful discussions, comments, and useful remarks regarding the results of this work.

{\small
\bibliographystyle{unsrtnat}
\bibliography{Samoilenko}

}

\end{document}